\pdfoutput=1
\begingroup\expandafter\expandafter\expandafter\endgroup
\expandafter\ifx\csname pdfsuppresswarningpagegroup\endcsname\relax
\else
\fi
\documentclass[a4paper,10pt]{article}
\usepackage[utf8]{inputenc}
\usepackage{tikz}
\usepackage{amsfonts,bm}
\usepackage{amsmath}
\usepackage{mathtools}
\usepackage{amssymb}
\PassOptionsToPackage{hyphens}{url}\usepackage[hidelinks]{hyperref}
\usepackage{graphicx}
\usepackage{tcolorbox}
\usepackage[numbers,sort&compress]{natbib}
\DeclareMathOperator*{\argmin}{arg\,min}
\DeclareMathOperator*{\argmax}{arg\,max}
\usepackage{float,xspace}
\usepackage{cprotect}
\usepackage[capitalise]{cleveref}
\usepackage{tcolorbox}
\usepackage[ruled,vlined,linesnumbered,algo2e]{algorithm2e}
\usepackage{float,subfig}
\usepackage{forloop}
\newcommand{\defcal}[1]{\expandafter\newcommand\csname c#1\endcsname{{\mathcal{#1}}}}
\newcommand{\defbb}[1]{\expandafter\newcommand\csname b#1\endcsname{{\mathbb{#1}}}}
\newcounter{calBbCounter}
\forLoop{1}{26}{calBbCounter}{
	\edef\letter{\Alph{calBbCounter}}
	\expandafter\defcal\letter
	\expandafter\defbb\letter
}
\usepackage[T1]{fontenc}    
\usepackage{booktabs}       
\usepackage{amsfonts}       
\usepackage{nicefrac}       
\usepackage{microtype}      

\usepackage{amsthm}
\newtheorem{theorem}{Theorem}

\newtheorem{observation}[theorem]{Observation}
\newtheorem{definition}[theorem]{Definition}
\newtheorem{lemma}[theorem]{Lemma}
\newtheorem{corollary}[theorem]{Corollary}

\newcommand{\Opt}{\text{\upshape{\texttt{OPT}}}\xspace}
\newcommand{\ksys}{k}
\newcommand{\maxcardinality}{\rho}
\DeclareMathOperator{\ground}{\cN}
\newcommand{\AlgStream}{{\textsc{\textsc{StreamingAlg}}}\xspace}
\newcommand{\AlgConstrained}{{\textsc{\textsc{ConstrainedAlg}}}\xspace}
\newcommand{\nnR}{{\bR_{\geq 0}}}
\crefname{algocf}{Algorithm}{Algorithms}
\Crefname{algocf}{Algorithm}{Algorithms}
\newcommand{\ceil}[1]{\lceil #1 \rceil}

\renewcommand{\emptyset}{\varnothing}

\usepackage[top=1.in, bottom=1in, left=1in, right=1in]{geometry}

\begin{document}
\title{Streaming Submodular Maximization under a $k$-Set System Constraint}
\author
{
Ran Haba\thanks{Depart. of Mathematics and Computer Science, The Open University of Israel. Email: \texttt{zook2005@gmail.com}.}
\and
Ehsan Kazemi\thanks{Yale Institute for Network Science, Yale University. Email: \texttt{ehsan.kazemi@yale.edu}.}
\and
Moran Feldman\thanks{Department of Computer Science, University of Haifa, Israel. Email: \texttt{moranfe@cs.haifa.ac.il}.}
\and
Amin Karbasi\thanks{Yale Institute for Network Science, Yale University. Email: \texttt{amin.karbasi@yale.edu}.}
}
\date{}

\maketitle
 
   \begin{abstract}
In this paper, we propose a novel framework that converts  streaming algorithms for  monotone submodular maximization into  streaming algorithms for non-monotone submodular maximization.  This reduction readily leads to the currently tightest deterministic approximation ratio for submodular maximization subject to a $k$-matchoid constraint. Moreover, we propose the first streaming algorithm for monotone submodular maximization subject to $k$-extendible and $k$-set system constraints.  Together with our proposed reduction, we obtain $O(k\log k)$ and $O(k^2\log k)$ approximation ratio for submodular maximization subject to the above constraints, respectively. We extensively evaluate the empirical performance of our algorithm against the existing work in a series of experiments including finding the maximum independent set in randomly generated graphs, maximizing linear functions over social networks, movie recommendation,  Yelp location summarization, and Twitter data summarization. 
 \end{abstract}

 \section{Introduction} \label{sec:introduction}
Submodularity  captures  an intuitive diminishing returns property where the benefit of an item decreases as the context in which it is considered grows. This property naturally occurs in many  applications where items may represent data points, features,  actions, etc. Moreover, submodularity is a sufficient condition that leads to an efficient optimization procedure for many discrete optimization problems. The above reasons have led to a surge of applications in machine learning where the gain of discrete choices shows diminishing returns and the optimization can be handled efficiently. Novel examples include  non-parametric learning~\citep{MKSK16}, dictionary learning \cite{das2011submodular}, crowd teaching \cite{singla2014near}, regression under human assistance \cite{de2019regression}, interpreting neural networks \cite{elenberg2017streaming}, adversarial attacks \cite{lei2019discrete}, data summarization \cite{dasgupta2013summarization,tschiatschek2014learning,elhamifar2017online,kirchhoff2014submodularity,kazemi2018scalable,mitrovic2018data}, fMRI  parcellation \citep{SKSC17}, and DNA sequencing \citep{LBN18}. 

More formally, a set function $f\colon 2^{\cN} \to \nnR$ is called \textbf{submodular}  if for all sets $A \subseteq B \subseteq \cN$ and element $u \notin B$ we have
\[
f(A \cup \{u\}) - f(A) \geq f(B \cup \{u\}) - f(B) \enspace.
\]
Moreover, a set function is called monotone if $f(A)\leq f(B)$ whenever $A\subseteq B$. The focus of this paper is on maximizing a general submodular function (not necessarily monotone). More concretely, we consider a very general form of constrained submodular maximization, i.e., 
\begin{equation}\label{eq:main-problem}
\Opt =  \argmax_{A\in \cI} f(A) \enspace,
\end{equation}
where $\cI$ represents the set of feasible solutions. For instance, when $f$ is monotone and $\cI$ represents a cardinality/size constraint,\footnote{Formally, $\cI$ contains in this case all subsets of $\cN$ of size at most $\maxcardinality$ for some value $\maxcardinality$.} the celebrated result of~\cite{nemhauser1978analysis} states that the greedy algorithm achieves a $(1-1/e)$-approximation for this problem, which is known to be optimal~\cite{nemhauser1978best}.  In the recent years, there has been a large body of literature aiming at solving Problem~\eqref{eq:main-problem} in the offline/centralized setting under various types of feasibility constraints such as matroid, $k$-matchoid, $k$-extendible system, and $k$-set system (formal definitions of some of these terms appear in \cref{sec:preliminaries}). These types of feasibility constraint (as well as other types not mentioned here) form an hierarchy, i.e., some types are generalization of other types. A part of this hierarchy is given by the following inclusions. Interestingly, all these inclusions are known to be strict. 
\begin{equation*}
\text{cardinality} \subset 
\text{matroid} \subset 
\text{intersection of $k$ matroids} \subset 
\text{$k$-matchoid} \subset
\text{$k$-extendible} \subset 
\text{$k$-set system} \enspace.
\end{equation*}

In the offline/centralized setting, the problem of maximizing a (non-monotone) submodular function subject to the above types of constraints is fairly well understood and easy-to-implement algorithms have been proposed. For instance, for maximization under a $k$-set system constraint, one obtains an approximation ratio of $k+O(\sqrt{k})$ using roughly $\sqrt{k}$ invocations of the natural greedy algorithm and an algorithm for unconstrained submodular maximization. Or, when the constraint is a $k$-extendible system, running the greedy algorithm only once over a carefully subsampled ground set achieves a $k+3$ approximation ratio~\cite{feldman2017greed}. It should also be noted that as is the greedy algorithm fails to provide any constant factor approximation guarantee when the submodular function is non-monotone, and thus, the above modifications of it are necessary.

In the streaming setting, when the elements arrive one at a time and the memory footprint is not allowed to grow significantly with the size of the data, the landscape of constrained submodular maximization is much less understood. In particular, even for the simple problem of monotone submodular maximization subject to a cardinality constraint,  the best known approximation guarantee is $1/2$~\cite{badanidiyuru2014streaming} (as opposed to $(1-1/e)$ in the offline setting). Moreover, no algorithm is currently known to achieve a non-trivial guarantee for more complicated constraints such as $k$-extendible or $k$-set system in the streaming setting even when the submodular objective function is monotone.

In this paper, we propose the first streaming algorithm for maximizing a general submodular function (not necessarily monotone) subject to a general $k$-set system constraint. Our algorithm achieves an $O(k^2\log k)$ approximation ratio for this problem. Moreover, when the constraint reduces to a $k$-extendible system the approximation guarantee of our streaming method improves to a better $O(k\log k)$ approximation ratio. Interestingly, the last approximation ratio is a significant improvement even compared to the best approximation ratio previously known for the very special case of this problem in which the objective function is linear. The current state-of-the-art algorithm for this special case, due to~\cite{crouch2014improved}, guarantees only an $O(k^2)$-approximation.

With the exception some algorithms designed for the simple cardinality constraint~\cite{alaluf2019making,badanidiyuru2014streaming,ene2019optimal,kazemi2019submodular}, all the streaming algorithms previously suggested for submodular maximization (see~\cite{buchbinder2014submodular,chekuri2015streaming,chakrabarti2015submodular,feldman2018doless}) have been based on the same basic technique. These algorithms maintain a feasible solution, and update it in the following way. When an element $u$ arrives, the algorithm (1) determines a set of elements that have to be removed from the current feasible solution to allow $u$ to be added without violating feasibility, and then (2) decides using some algorithm specific rule whether it is beneficial to make this trade (i.e., add $u$ and remove the necessary elements to recover feasibility). Our algorithm uses a very different technique of maintaining multiple feasible solutions to which elements can be added (but can never be removed), which is inspired by the technique of~\cite{crouch2014improved} for maximization of linear functions subject to $k$-set systems. Intuitively, each one of the solutions maintained by our algorithm is associated with a particular importance of elements, and the role of this solution is to collect enough elements of this importance. Since we collect elements from each level of importance, once the stream ends, the union of the solutions we maintain is a good enough summary of the stream, and our algorithm is able to pick a feasible subset of this union which is competitive with respect to the optimal solution.

One component of our algorithm is a general framework that is able to convert many streaming algorithms for monotone submodular maximization to similar algorithms for non-monotone submodular maximization. As an immediate consequence of this framework, we get a deterministic streaming algorithm for maximizing a general (not necessarily monotone) submodular function subject to a $k$-matchoid constraint, which is a slight improvement over the state-of-the-art deterministic approximation ratio for this problem due to~\cite{chekuri2015streaming}. We also compare the empirical performance of our algorithm with the existing work and natural baselines in a set of experiments including independent set over randomly generated graphs, maximizing a linear function over edges of a graph, movie recommendation, and Yelp location data summarization. In all these applications, the various constraints are modeled as an instance of a $k$-set system.

Before concluding this section, we need to highlight a technical issue. The standard definition of streaming algorithms requires them to use poly-logarithmic amount of space, which is less than the space necessary for keeping a solution for our problem. Thus, no algorithm for this problem aiming to produce a solution (rather than just estimate the value of the optimal solution) can be a true streaming algorithm. This is true also for all the above mentioned streaming algorithms, which are in fact semi-streaming algorithms---a semi-streaming algorithm is an algorithm that processes the data as a sequence of elements using an amount of space which is nearly linear in the maximum size of a feasible solution and typically makes only a single pass over the entire data stream. Since true streaming algorithms are almost irrelevant to our setting, we ignore the distinction between streaming and semi-streaming algorithms in this paper and often use the term ``streaming algorithm'' to refer to a semi-streaming algorithm.

\paragraph{Paper Structure.} In \cref{sec:preliminaries}, we formally define some types of constraints and the notation we use, and then formally state some technical results that we need. In \cref{sec:framework}, we describe our above mentioned framework for converting streaming algorithms for monotone submodular maximization into streaming algorithms for non-monotone submodular maximization. Then, in \cref{sec:algorithm}, we describe and formally analyze our algorithm, and in \cref{sec:experiments} we describe the experiments we conducted to study the empirical performance of this algorithm.

\subsection{Related Work}

The study of submodular maximization in the streaming setting was initialized by the works of \citet{badanidiyuru2014streaming} and \citet{chakrabarti2015submodular}. As discussed above, the work of~\cite{chakrabarti2015submodular} was based on a technique allowing the removal of elements from the solution (also known as preemption). Originally,~\cite{chakrabarti2015submodular} suggested this technique only for constraints formed by the intersection of $k$-matroids and a monotone submodular objective function, but later works extended the use of the technique to the more general class of $k$-matchoid constraints as well as non-monotone submodular functions~\cite{buchbinder2014submodular,chekuri2015streaming,feldman2018doless}. The above mentioned algorithm of \citet{badanidiyuru2014streaming} works only for the simple cardinality constraint and monotone submodular objective functions, but provides an improved approximation ratio of $\nicefrac{1}{2}$ for this setting (there is evidence that this approximation ratio is optimal for the setting~\cite{norouzifard2018beyond}). The technique at the heart of this algorithm is based on growing a set to which elements can only be added, which becomes the output solution of the algorithm by the end of the stream (unlike the case in the technique of~\cite{crouch2014improved} on which we base our results, in which the final solution is obtained by combining multiple sets grown by the algorithm). More recent works improved the algorithm of~\cite{badanidiyuru2014streaming} by improving its space complexity~\cite{kazemi2019submodular} and extending its technique to non-monotone submodular functions~\cite{alaluf2019making,ene2019optimal}.

The study of submodular maximization in the offline/centralized setting is very vast, and thus, we concentrate here only on results for general $k$-extendible or $k$-set system constraints. Already in $1978$, \citet{fisher1978analysis} proved that the natural greedy algorithm obtains $k + 1$ approximation for the problem of maximizing a monotone submodular function subject to a $k$-set system constraint (some of their proof was given implicitly, and the details were filled in by~\cite{calinescu2011maximizing}). This was recently proved to be almost optimal. Specifically, \citet{badanidiyuru2014fast} proved that no polynomial time algorithm can obtain $k - \varepsilon$ approximation for this problem for any constant $\varepsilon > 0$, and the same inapproximability result was later shown to apply also to $k$-extendible constraints by~\cite{feldman2017greed}. As mentioned in \cref{sec:introduction}, \citet{feldman2017greed} presented the state-of-the-art algorithms for maximizing a (not necessarily monotone) submodular function subject to $k$-set system and $k$-extendible constraints. Both algorithms obtain $k + o(k)$ approximation, which improves over two previous results due to~\cite{gupta2010constrained} and~\cite{mirzasoleiman2016fast} that obtained roughly $3k$ and $2k$ approximation, respectively, for the more general case of a $k$-set system constraint.
\section{Preliminaries and Notation} \label{sec:preliminaries}

We begin this section by presenting some notation that we use in this paper. Then, we formally define some types of constraints mentioned in \cref{sec:introduction}, and discuss the guarantee of a simple greedy algorithm for these constraints.

Given an element $u$ and a set $A$, we use $A + u$ as a shorthand for the union $A \cup \{u\}$. We also denote the marginal gain of adding $u$ to $A$ with respect to a set function $f\colon 2^\cN \to \bR$ using $f(u \mid A) \triangleq f(A + u) - f(A)$. Similarly, the marginal gain of adding a set $B \subseteq \cN$ to another set $A \subseteq \cN$ is denoted by $f( B \mid A) \triangleq f(B \cup A) - f(A)$. Note that this notation allows us, for example, to rewrite the definition of submodularity as the requirement that $f(u \mid A) \geq f(u \mid B)$ for every two sets $A \subseteq B \subseteq \cN$ and element $u \not \in B$.

A constraint is defined, for our purposes, as a pair $(\cN, \cI)$, where $\cN$ is a ground set and $\cI$ is the collection of all feasible subsets of $\cN$. All the types of constraints discussed in \cref{sec:introduction} are independence systems according to the following definition.
\begin{definition}
	Given a ground set $\cN$ and a collection of sets $\cI \subseteq 2^\cN$, the pair $(\cN, \cI)$ is an independence system if (i) $\varnothing \in \cI$ and (ii) for $B \in \cI$ and any $A \subseteq B$ we have $A \in \cI$.
\end{definition}
It is customary to call a set $A \subseteq \cN$ \emph{independent} if it belongs to $\cI$ and \emph{dependent} if it does not (i.e., it is infeasible). 
An independent set $B \in \cI$ which is maximal with respect to inclusion is called a \emph{base}; that 
is, 
$B \in \cI$ is a base if $A \in \cI$ and $B \subseteq A$ imply that $B = A$. Furthermore, an independent set $B \in \cI$ which is a subset of some set $E \subseteq \cN$ is called a base of $E$ if it is a base of the independence system $(E, 2^E \cap \cI)$. Note that this means that a set $B$ is a base of $(\cN, \cI)$ if and only if it is a base of $\cN$.

The above terminology allows us now to define $k$-set systems.
\begin{definition} \label{def:k-set-system}
	An independence system $(\cN, \cI)$ is a $k$-set system for an integer $k \geq 1$ if for every set $E \subseteq \cN$, all the bases of $E$ have the same size up to a factor of $k$ (in other words, the ratio between the sizes of the largest and smallest bases of $E$ is at most $k$).
\end{definition}

An immediate consequence of the definition of $k$-set systems is that any base of such a system
is a maximum size independent set up to an approximation ratio of $k$. Thus, one can get a $k$-approximation for the problem of finding a maximum size set subject to a $k$-set system constraint by outputting an arbitrary base of the $k$-set system, which can be done using the following simple strategy. Start with the empty solution, and consider the elements of the ground set $\cN$ in an arbitrary order. When considering an element, add it to the current solution, unless this will make the solution dependent. We refer to this procedure as the unweighted greedy algorithm.

Let us now define $k$-extendible systems. We remind the reader that $k$-extendible systems are well-known to be a restricted class of $k$-set systems.

\begin{definition} \label{def:k-extendible-system}
	An independence system $(\cN, \cI)$ is a $k$-extendible system for an integer $k \geq 1$ if for any two independent sets $S \subseteq T \subseteq \cN$, and an element $u \not \in T$ such that $S + u \in \cI$, there is a subset $Y \subseteq T \setminus S$ of size at most $k$ such that $T \setminus Y + u \in \cI$.
\end{definition}



Since $k$-extendible systems are, in particular, $k$-systems, the above discussion already implies that the unweighted greedy algorithm obtains $k$-approximation for the problem of finding a maximum size independent set in such a system. The following lemma strengthens this observation, and is the key technical reason that our algorithm has a better approximation guarantee for $k$-extendible system constraints than for $k$-set system constraints.

\begin{lemma}
	\label{lem:greedy_improved_result}
	Given a $k$-extendible set system $(\cN, \cI)$, the unweighted greedy algorithm is guaranteed to produce an independent set $B$ such that $k \cdot |B \setminus A| \geq |A \setminus B|$ for any independent set $A \in \cI$.
\end{lemma}

\begin{proof}
	Let us denote the elements of $B \setminus A$ by $x_1,x_2,\dotsc,x_m$ in an arbitrary order. Using these elements, we recursively define a series of independent sets $A_0, A_1,\dotsc,A_m$. The set $A_0$ is simply the set $A$. For $1 \leq i \leq m$, we define $A_i$ using $A_{i - 1}$ as follows. Since $(\cN, \cI)$ is a $k$-extendible system and the subsets $A_{i-1}$ and $A_{i-1} \cap B + x_i \subseteq B$ are both independent, there must exist a subset $Y_i \subseteq A_{i-1} \setminus (A_{i-1} \cap B) = A_{i-1} \setminus B$ such that $|Y_i| \leq k$ and $A_{i-1} \setminus Y_i+x_i\in \cI$. Using the subset $Y_i$, we now define $A_i = A_{i-1} \setminus Y_i+x_i$.
	Note that by the definition of $Y_i$, $A_i \in \mathcal{I}$ as promised.
	Furthermore, since $Y_i \cap B = \emptyset$ for each $0 \leq i \leq m$, we know that $(A \cup \{x_1, x_2, \dots, x_m\}) \cap B \subseteq A_m$, which implies $B \subseteq A_m$ because $\{x_1,x_2,\dotsc,x_m\} = B \setminus A$.
	However, $B$, as the output of the unweighted greedy algorithm, must be an inclusion-wise maximal independent set (i.e., a base), and thus, it must be in fact equal to the independent set $A_m$ containing it.
	
	Let us now denote $Y = \bigcup_{i=1}^{m} Y_i$, and consider two different ways to bound the number of elements in $Y$. On the one hand, since every set $Y_i$ includes up to $k$ elements, we get $|Y| \leq km = k \cdot |B \setminus A|$. On the other hand, the fact that $B = A_m$ implies that every element of $A \setminus B$ belongs to $Y_i$ for some value of $i$, and therefore, $|Y| \geq |A \setminus B|$. The lemma now follows by combining these two bounds. 
\end{proof}
\section{Streaming Algorithms for Non-monotone Submodular Maximization} \label{sec:framework}

\citet{mirzasoleiman2018streaming} proposed a framework for the following task. Given a streaming\footnote{Recall that in this paper we use the term ``streaming algorithm'' to refer to algorithms that are technically ``semi-streaming algorithms'', i.e., their space complexity is allowed to be nearly-linear in the size of the output set.} algorithm for maximizing monotone submodular functions, the framework produces a similar algorithm that works also for non-monotone submodular objectives. Unfortunately, however, this framework applies only to algorithms satisfying a property which, to the best of our knowledge, is not satisfied by any streaming algorithm from the literature (except algorithms that work for non-monotone functions by design). In particular, this is the case for the algorithm of \citet{chekuri2015streaming} explicitly mentioned by~\cite{mirzasoleiman2018streaming} as a natural fit for their framework. In the rest of this section we discuss this issue in more detail, and then introduce a different framework which achieves the same goal (converting algorithms for monotone submodular maximization into algorithms for non-monotone submodular maximization), but requires a different property from the input algorithms which is satisfied by both existing algorithms from the literature and the new algorithm we suggest in this paper.

The algorithm of~\cite{chekuri2015streaming} discussed above is a streaming algorithm for maximizing monotone submodular functions under a $k$-matchoid constraint, and \citet{mirzasoleiman2018streaming} applied their framework to it in order to get such an algorithm for non-monotone function. Formally, this framework requires the input streaming algorithm to satisfy the inequality
 \begin{align} \label{eq:prev-req-cond}
  f(S) \geq \alpha \cdot f(S \cup T) \enspace,
 \end{align}
where $S$ as the output of the algorithm, $T$ is an arbitrary feasible solution and $\alpha$ is a positive value. Unfortunately, the algorithm of~\cite{chekuri2015streaming} fails to satisfy \cref{eq:prev-req-cond} for any constant $\alpha$, so does the algorithms of \citet{buchbinder2019online} and \citet{chakrabarti2015submodular}. In \cref{appendix:counter-example} we provide examples showing that this is the case for all these algorithms even under a simple cardinality constraint. 

Interestingly, \citet{chekuri2015streaming} presented, prior to the work of~\cite{mirzasoleiman2018streaming}, an alternative method to convert their algorithm into a deterministic algorithm for non-monotone functions based on a technique due to~\citet{gupta2010constrained}. The framework we suggest can be viewed as a formalization and generalization of this technique. As an alternative to the property~\eqref{eq:prev-req-cond}, our framework uses the property described by Definition~\ref{def:req-cond}. 
\begin{definition} \label{def:req-cond}
Consider a data stream algorithm for maximizing a non-negative submodular function $f\colon 2^\cN \to \nnR$ subject to a constraint $(\cN, \cI)$. We say that such an algorithm is an $(\alpha, \gamma)$-approximation algorithm, for some $\alpha \geq 1$ and $\gamma \geq 0$,
if it returns two sets $S \subseteq A \subseteq \cN$ such that $S \in \cI$, and for all $T \in \cI$ we have
\[
	\bE[f(T \cup A)] \leq \alpha \cdot \bE[f(S)] + \gamma \enspace.
\]
\end{definition}

We note that most previous algorithms (including the algorithms suggested by~\cite{buchbinder2019online,chakrabarti2015submodular}) satisfy Definition~\ref{def:req-cond} with $\gamma = 0$. However, they do not keep in memory the set $A$ because this set can get very large. This is unacceptable for us, as we need explicit access to $A$. Fortunately, \citet{chekuri2015streaming} described a technique to create a tradeoff between the size of $A$ and the value of $\gamma$, and by setting the parameters right it is possible to keep both $\gamma$ and $|A|$ reasonably small. The same technique can be used to get a similar result for the algorithms of~\cite{buchbinder2019online,chakrabarti2015submodular} as well.

We are now ready to describe the algorithm at the heart of our framework. This algorithm is given as Algorithm~\ref{alg:non-monotone}, and it assumes access to two procedures: (1) a data stream algorithm $\AlgStream$ for the problem of maximizing a non-negative submodular function $f\colon 2^\cN \to \nnR$ subject to a constraint $(\cN, \cI)$, and (2) an offline algorithm $\AlgConstrained$ for the same problem. For the data stream algorithm $\AlgStream$ we use the following, not very standard, semantics. Algorithm~\ref{alg:non-monotone} has two ways to call $\AlgStream$. Every time that Algorithm~\ref{alg:non-monotone} would like to pass additional elements to $\AlgStream$, it calls it with the set of these new elements, and $\AlgStream$ updates its internal data structures accordingly and returns a set including all the elements that it decided to remove from its memory. Once the stream ends,   Algorithm~\ref{alg:non-monotone} calls $\AlgStream$ with the subscript $\mathsf{end}$ and pass to it any final elements it would like $\AlgStream$ to get. $\AlgStream$ then process these elements and returns three sets: the sets $S$ and $A$ produced by $\AlgStream$ (as described in Definition~\ref{def:req-cond}) and a set $D$ consisting of all the elements that are still in the memory of $\AlgStream$ and did not end up in $A$. Figure~\ref{fig:alg_schematic} is a graphic representation of the flow of elements between the components of \cref{alg:non-monotone}.

\begin{algorithm2e}
	\DontPrintSemicolon
	\caption{Non-monotone Data Stream Algorithm} \label{alg:non-monotone}
	\textbf{Input: } a positive integer $r$\\
	\textbf{Output: } a set $S \in \cI$\\
	Initialize $r$ independent copies of $\AlgStream$: $\AlgStream^{(1)}, \dotsc, \AlgStream^{(r)}$.\\
	\While{there are more elements in the stream}
	{
		Let $D_0$ be a singleton set containing the next element of the stream.\\
		\lFor{$i = 1$ \KwTo $r$}
		{
			$D_i \gets \AlgStream^{(i)}(D_{i-1})$.
		}
	}
	Let $D_0 \gets \varnothing$.\\
	\For{$i = 1$ \KwTo $r$}
	{
		$[S_i, A_i, D_i] \gets \AlgStream^{(i)}_{\mathsf{end}}(D_{i-1})$.\\
		$S'_i \gets \AlgConstrained(A_i)$. \\
	}
	\Return{the set maximizing $f$ among $\{S_i, S'_i\}_{i = 1}^r$}.
\end{algorithm2e}

\begin{figure}[ht!] 
	\centering
	\includegraphics[width=4.5in]{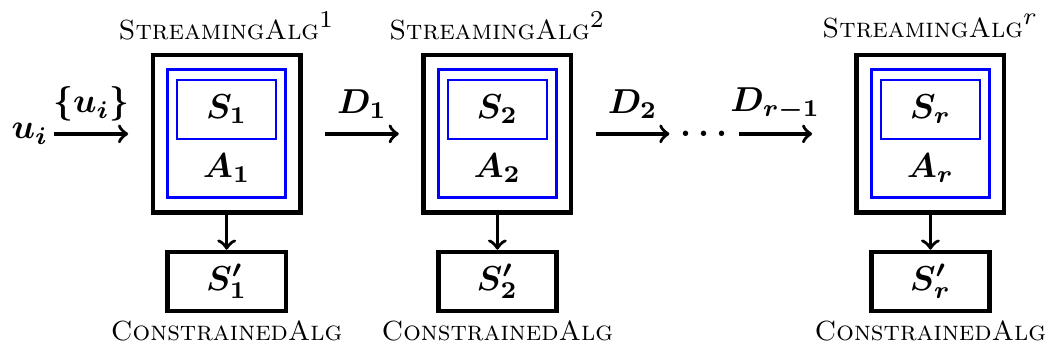}
	\caption{Schematic representation of Algorithm~\ref{alg:non-monotone}.}\label{fig:alg_schematic}
\end{figure}

It is clear that Algorithm~\ref{alg:non-monotone} outputs a feasible solution. The following observation bounds the space complexity of Algorithm~\ref{alg:non-monotone}.
\begin{observation}
The space complexity of Algorithm~\ref{alg:non-monotone} is upper bounded by $O(r \cdot M_{\AlgStream} + M_{\AlgConstrained})$, where $M_{\AlgStream}$ and $M_{\AlgConstrained}$ represent the space complexities of their matching algorithms under the assumption that the input for $\AlgStream$ is a subset of the full input and the input for $\AlgConstrained$ is the $A$ set produced by $\AlgStream$ on some such subset.
\end{observation}
\begin{proof}
Note that every set assigned to a variable $D_i$ by Algorithm~\ref{alg:non-monotone} is either an input set for some copy of $\AlgStream$ or an output set of such a copy. Furthermore, any such set is either ignored or fed immediately after construction to some copy of $\AlgStream$. Thus, the space complexity required for these sets is upper bounded by the space complexity required for the $r$ copies of \AlgStream used by Algorithm~\ref{alg:non-monotone}, which is $O(r \cdot M_{\AlgStream})$. Additionally, since the sets $S_i$ and $A_i$ are the final outputs of these copies of $\AlgStream$, we get they can also be stored using $O(r \cdot M_{\AlgStream})$ space. Finally, every set $S'_i$ is a subset of $A_i$, and therefore, the sets $S'_i$ do not require more space than the sets $A_i$. Combining all the above, we get that the space complexity of Algorithm~\ref{alg:non-monotone}---excluding the space required for running \AlgConstrained---is at most $O(r \cdot M_{\AlgStream})$.
\end{proof}

To complete the analysis of Algorithm~\ref{alg:non-monotone}, it remains to analyze its approximation guarantee. Towards this goal, we need the following known lemma.
\begin{lemma}[Lemma 2.2 of~\citep{buchbinder2014submodular}] \label{lem:buchbinder}
	Let $g\colon 2^\cN \to \nnR$ be a non-negative submodular function, and let $B$ be a 
	random 
	subset of $\cN$ containing every element of $\cN$ with probability at most $q$ 
	(not 
	necessarily 
	independently). Then,
	$\bE[g(B)] \geq (1 - q) \cdot g(\varnothing)$.
\end{lemma}

\begin{lemma} \label{lem:reduction_ratio}
	Assume \AlgStream is an $(\alpha, \gamma)$-approximation algorithm and \AlgConstrained is an offline $\beta$-approximation algorithm. Then, Algorithm~\ref{alg:non-monotone} returns a solutions $S$ such that
	\[ 	\bE[f(S)] \geq \dfrac{(r-1) \cdot \Opt - r\gamma}{r\alpha +  r(r-1)\beta/2}  \enspace. \]
\end{lemma}
\begin{proof}
Let $S^*$ be an arbitrary optimal solution, i.e., a set obeying $S^* \in \cI$ and $f(S^*) = \Opt$. 
For every integer $1 \leq i \leq r$, we denote by $\cN_i$ the set of elements that $\AlgStream^{(i)}$ has received. 
Note that we have $\cN_1 = \cN$ and $\cN_i = \cN \setminus (\cup_{1 \leq j \leq i-1} A_j)$ for $2 \leq i \leq r$ since $\AlgStream^{(i)}$ outputs every element that it gets and does not end up in $A_i$ as an element of $D_i$ at some point. Since $A_i$ is a subset of $\cN_i$, this implies that the sets $A_1, A_2, \dotsc, A_r$ are disjoint. Let us also define $S^*_i = S^* \cap \cN_i$.

Let us define now $\bar{A}$ to be a uniformly random set from $\{ A_1, A_2, \dotsc, A_r\}$, and $g(S) = f(S \cup S^{*})$. Then,
\begin{align} \label{eq:bound-average}
	\dfrac{1}{r} \sum_{i= 1}^{r} f(A_i \cup S^*) = \bE_{\bar{A}}[f(\bar{A} \cup S^*)] =  \bE_{\bar{A}}[g(\bar{A})] \geq \left(1 - \dfrac{1}{r}\right) \cdot g(\emptyset) = \left(1 - \dfrac{1}{r}\right) \cdot f(S^*) \enspace,
\end{align}
where the notation $\bE_{\bar{A}}$ stands for expectation over the random choice of $\bar{A}$ out of $\{ A_1, A_2, \dotsc, A_r\}$ (but not over any randomness that might be introduced by \AlgStream), and the inequality results from Lemma~\ref{lem:buchbinder} because (i) every element of $\cN$ belongs to $\bar{A}$ with probability at most $\frac{1}{r}$ since the sets $A_i$ are disjoint, and (ii) $g$ is a non-negative submodular function on its own right.

Note that $S_i$ contains exactly the elements of $S^*$ that do not appear in any of the sets $A_1, A_2, \dotsc, A_{i - 1}$. Thus, $S^* = S_i \cup \left(\cup_{1 \leq j \leq i-1} \left(S^* \cap  A_j \right) \right)$. By the submodularity and non-negativity of $f$, this implies
\[ f(A_i \cup S^*_i) + f(\cup_{j=1}^{i-1} (A_j \cap S^* )) \geq f(A_i \cup S^*) \enspace. \]
Using the last inequality, we can write
\begin{align*}
	(r - 1) \cdot \Opt
	={} &
	(r-1) \cdot f(S^*)
	\leq
	\sum_{i= 1}^{r}f(A_i \cup S^*)\\
	\leq{} &
	\sum_{i= 1}^{r} f(A_i \cup S^*_i) + f(\cup_{j=1}^{i-1} (A_j \cap S^* ))
	\leq
	\sum_{i= 1}^{r} \left[  f(A_i \cup S^*_i) + \sum_{j=1}^{i-1} f(A_j \cap S^*) \right]
	\enspace,
\end{align*}
where the first inequality follows from Inequality~\eqref{eq:bound-average}. Taking now expectation over any randomness introduced by $\AlgStream$ and $\AlgConstrained$, we get from the last inequality using the guarantees of these two algorithms that
\begin{align*}
	(r - 1) \cdot \Opt
	\leq{} &
	\sum_{i= 1}^{r} \left[  \bE[f(A_i \cup S^*_i)] + \sum_{j=1}^{i-1} \bE[f(A_j \cap S^*)] \right]\\
	\leq{} &
	\sum_{i= 1}^{r} \left[ \alpha \cdot \bE[f(S_i)] + \gamma + \sum_{j=1}^{i-1} \beta \cdot \bE[f(S'_j)] \right]
	\leq
	\left[ r\alpha + \dfrac{\beta r(r-1)}{2} \right] \cdot \bE[f(S)] + r\gamma
	\enspace,
\end{align*}
where the last inequality holds since $S$ is selected as the set maximizing $f$ among all the sets $\{S_i, S'_i\}_{i = 1}^r$.
\end{proof}

The following theorem summarizes the results we have proved in this section.

\begin{theorem}\label{theorem:non-monotone-streaming}
Given an $(\alpha, \gamma)$-approximation data stream algorithm $\AlgStream$ for maximizing a non-negative submodular function subject to some constraint and an offline $\beta$-approximation algorithm $\AlgConstrained$ for the same problem. There exists a data stream algorithm returning a feasible set $S$ that obeys
	\[ 	\bE[f(S)] \geq \dfrac{(r-1) \cdot \Opt - r\gamma}{r\alpha +  r(r-1)\beta/2}  \enspace. \]
Furthermore,
\begin{itemize}
	\item this algorithm is deterministic if \AlgStream and \AlgConstrained are both deterministic.
	\item the space complexity of this algorithm is upper bounded by $O(r \cdot M_{\AlgStream} + M_{\AlgConstrained})$, where $M_{\AlgStream}$ and $M_{\AlgConstrained}$ represent the space complexities of their matching algorithms under the assumption that the input for $\AlgStream$ is a subset of the full input and the input for $\AlgConstrained$ is the $A$ set produced by $\AlgStream$ on some such subset.
\end{itemize}
\end{theorem}

We note that the algorithm guaranteed by \cref{theorem:non-monotone-streaming} is a streaming algorithm when the algorithm $\AlgStream$ is a streaming algorithm, the algorithm \AlgConstrained is a nearly-linear space algorithm and $r$ is upper bounded by a poly-log function.

In Appendix~\ref{sec:application} we show that by plugging one of the versions of the algorithm of \citet{chekuri2015streaming} into our framework it is straightforward to get a deterministic streaming algorithm for the problem of maximizing a non-negative (not necessarily monotone) submodular function subject to a $k$-matchoid constraint whose approximation ratio is $(15/2 + 4\varepsilon) k + O(\sqrt{k})$ for every constant $\varepsilon > 0$, which is a slight improvement over the guarantee of the previous state-of-the-art deterministic streaming algorithm for this problem (also due to~\cite{chekuri2015streaming}) which has an approximation guarantee of $8k + \gamma$, where $\gamma$ is the approximation ratio of the best offline algorithm for the same problem.
\section{Streaming Algorithm for \texorpdfstring{$k$}{k}-System and \texorpdfstring{$k$}{k}-Extendible System Constraints} \label{sec:algorithm}

In this section we formally prove our results for the problem of maximizing a non-negative submodular function $f\colon 2^\cN \to \nnR$ subject to a $k$-system or $k$-extendible system constraint $(\cN, \cI)$. A simple version of the algorithm we use to prove these results is given as Algorithm~\ref{alg:streaming-k-system}. This version assumes pre-access to a value $\rho$ equal to the size of the largest independent set in $\cI$ and a threshold $\tau$ estimating the value $M = \max_{u \in \cN, \{u\} \in \cI} f(\{u\})$. In Appendix~\ref{app:full_algorithm}, we present a more involved version of our algorithm that does not need this pre-access, but has a space complexity larger than the space complexity of Algorithm~\ref{alg:streaming-k-system} by a factor of $O(\log \rho + \log k)$---the approximation guarantee remains unchanged.

Intuitively, \cref{alg:streaming-k-system} maintains $\ell$ independent sets $E_i$, where each one of these sets corresponds to a different range of marginal contributions: the larger $i$, the smaller the marginal contributions $E_i$ is associated with. When an element $u$ arrives, the algorithm calculates the marginal contribution $m(u)$ of $u$ with respect to the union of the $E_i$ sets,\footnote{The notation $m(u)$ might suggest that the value $m(u)$ depends only on the identity of the element $u$. However, this is not the case. In fact, $m(u)$ might depend also on the set of elements that arrived before $u$ and the order of their arrival.} and then adds $u$ to the $E_i$ corresponding to this marginal contribution, unless this violates the independence of this $E_i$. Once the entire input has been processed, the algorithm combines the $E_i$ sets into $h$ possible output sets $T_0, T_1, \dotsc, T_{h - 1}$. Each output sets $T_j$ is constructed by greedily taking elements from the $E_i$ sets obeying $i \equiv j \pmod{h}$ (the algorithm scans the sets obeying this condition in an increasing $i$ order, which is a decreasing order with respect to the marginal contributions associated with these sets). The final output of the algorithm is simply the best set among the sets $T_0, T_1, \dotsc, T_{h - 1}$.

\begin{algorithm2e}
	\DontPrintSemicolon
	\caption{Streaming Algorithm for $\ksys$-Systems} \label{alg:streaming-k-system}
	\textbf{Input: } a threshold $\tau \in [M, 2M]$, the size $\maxcardinality$ of the largest independent set, and the parameter $k$ of the constraint.\\
	\textbf{Output: } a solution $T \in \cI$\\
	Let $\ell \gets \lfloor \log_2 (4\maxcardinality) \rfloor$ and $h \gets \lceil \log_2(2k+1) \rceil$. \\
	\lFor{$i = 0$ \textbf{to} $\ell$}{Initialize $E_i \gets \emptyset$.}
	\For{every element $u$ arriving}
	{
	Let $m(u) \leftarrow f\left(u \mid \cup_{i=0}^{\ell} E_{i}\right) $.\\
	\leIf{$m(u) > 0$}{Let $i(u) \leftarrow\left\lfloor\log _{2}(\tau / m(u))\right\rfloor$}{Let $i(u) \gets \infty$.}
	\lIf{$0 \leq i(u) \leq \ell $ and $E_{i(u)} + u \in \cI$}{Update $E_{i(u)} \gets E_{i(u)} + u$.}
}
\BlankLine
\For{$j=0$ \textbf{to} $h - 1$}
{
	Let $i \gets j$ and $T_j \gets \emptyset$.\\
	\While{$i \leq \ell$}
	{
		\lWhile{there is an element $u \in E_i$ such that $T_j + u \in \cI$}{Update $T_j \gets T_j + u$.}
		$i \gets i + h$.
}
}
	\Return{the set $T$ maximizing $f$ among $T_0, T_1, \cdots, T_{h - 1}$}.
\end{algorithm2e}

We begin the analysis of \cref{alg:streaming-k-system} by showing that it has the space complexity of a semi-streaming algorithm.
\begin{lemma}
\cref{alg:streaming-k-system} stores $O(\maxcardinality (\log \maxcardinality + \log k)) = \tilde{O}(\maxcardinality)$ elements at every given time point.
\end{lemma}
\begin{proof}
\cref{alg:streaming-k-system} stores elements only in the sets $E_0, E_1, \dotsc, E_\ell$ and the sets $T_0, T_1, \dotsc, T_{h - 1}$. Since these sets are kept independent by the algorithm, each one them contains at most $\maxcardinality$ elements. Thus, the number of elements stored by \cref{alg:streaming-k-system} is upper bounded by
\[
	(\ell + h)\maxcardinality
	=
	[O(\log \maxcardinality + \log k)]\maxcardinality
	=
	O(\maxcardinality (\log \maxcardinality + \log k))
	\enspace.
	\qedhere
\]
\end{proof}

Our next objective to analyze the approximation ratio of \cref{alg:streaming-k-system}. The majority of the work in this analysis is showing that the value of the output set of the algorithm is proportional to $f(E)$, where $E = \cup_{i = 0}^\ell E_i$. However, for such a guarantee to be useful, we first need to show that $f(E)$ is large.
\begin{lemma} \label{lem:E_bound}
For every set $S \in \cI$, $f(E \mid \varnothing) = \sum_{i=0}^{\ell} \sum_{u \in E_i} m(u) \geq \frac{f(S \cup E \mid \varnothing) - \tau/4}{2k+1}$.
\end{lemma}
\begin{proof}
	First, note that we have $f(E \mid \varnothing) = \sum_{i=0}^{\ell} \sum_{u \in E_i} m(u)$ because $m(u)$ is the marginal contribution of $u$ with respect to the elements that were added to $\cup_{i = 0}^\ell E_i$ before $u$. Let us also define, for every integer $0 \leq i \leq \ell $, $S_i = \{ u \in S \mid i(u) = i \}$. Then,
	\begin{align*}
	f(E \mid \varnothing) = \sum_{i=0}^{\ell} \sum_{u \in E_i} m(u)  & \geq \sum_{i=0}^{\ell} |E_i| \cdot  \dfrac{\tau}{2^{i+1}} 
	\geq \dfrac{1}{k} \cdot \sum_{i=0}^{\ell} |S_i| \cdot \dfrac{\tau}{2^{i+1}} \\
	&\geq \dfrac{1}{2k} \cdot \sum_{i=0}^{\ell} \sum_{u \in S_i} m(u)
	= \dfrac{1}{2k} \cdot \left[\sum_{u \in S} m(u) - \sum_{\substack{u \in S \\ i(u) < 0 \text{ or } i(u) > \ell}} \mspace{-36mu}  m(u) \right] 
	\enspace,
	\end{align*}
	where the first and third inequalities hold since an element $u$ is added to a set $E_i$ only when $i = i(u)$, and the second inequality holds since one can view $E_i$ as the output of running the unweighted greedy algorithm on a ground set which includes the independent set $S_i$ as a subset.
	
	By the submodularity of $f$, we can immediately get
	\[
		\sum_{u \in S} m(u)
		\geq
		\sum_{u \in S} f(u \mid E)
		\geq
		f(S \mid E)
		=
		f(S \cup E \mid \varnothing) - f(E \mid \varnothing)
		\enspace.
	\]
	We also note that $\{u\} \in \cI$ for every element $u \in S$ because $S$ itself is independent, and thus, $\tau \geq M \geq f(\{u\}) \geq m(u)$ (recall that $M$ was defined as $\max_{u \in \cN, \{u\} \in \cI} f(\{u\})$). Hence, $i(u) = \lfloor \log_2(\tau / m(u)) \rfloor \geq 0$, which implies
	\[
		\sum_{\substack{u \in S \\ i(u) < 0 \text{ or } i(u) > \ell}} \mspace{-36mu}  m(u)
		=
		\sum_{\substack{u \in S \\ i(u) > \ell}} m(u)
		\leq
		\sum_{\substack{u \in S \\ i(u) > \ell}} \frac{\tau}{2^{\ell + 1}}
		\leq
		\maxcardinality \cdot \frac{\tau}{2^{\log_2 (4\maxcardinality)}}
		=
		\frac{\tau}{4}
		\enspace.
	\]
	Combining all the above inequalities gives us
	\[
		f(E \mid \varnothing)
		\geq
		\frac{1}{2k} \cdot \left[f(S \cup E \mid \varnothing) - f(E \mid \varnothing) - \frac{\tau}{4} \right] 
		\enspace,
	\]
	and the lemma follows by rearranging this inequality.
\end{proof}

As discussed above, our next objective is to relate the value of the output set of \cref{alg:streaming-k-system} to $f(E)$. As an intermediate step, we relate $f(T_j)$ to the sum of the $m(u)$ values of the elements $u$ that belong to the sets $E_i$ that are combined to create $T_j$ (recall that these are exactly the sets $E_i$ for which $i \equiv j \pmod{h}$). In the next lemma we assume that the constraint $(\cN, \cI)$ is a $k$-system. Naturally, the lemma holds also for constraints that are $k$-extendible systems, but for such constraints it is possible to get a better bound on $f(T_j)$ using a more careful analysis, and this bound appears below as Lemma~\ref{lem:T_extendible}.

Intuitively, the next lemma holds because when \cref{alg:streaming-k-system} adds elements of a set $E_i$ to a set $T_j$, this increases the size of the set $T_j$ to at least $E_i/k$ (since the constraint is $k$-set system). Thus, either about $1/k$ of the elements of $E_i$ are added to $T_j$, or the size of $T_j$ before the addition of the elements of $E_i$ is already significant compared to the size of $E_i$. Moreover, in the later case, the elements of $T_j$ can pay for the elements of $E_i$ that they have blocked because they all have a relatively high value (as they originate in a set $E_{i'}$ for some $i' \leq i - h$).

\begin{lemma} \label{lem:T_system}
	If $(\cN, \cI)$ is a $k$-set system, then for every integer $0 \leq j < h$ we have
	\[
		f(T_j \mid \varnothing) \geq \frac{1}{4k} \cdot \mspace{-18mu} \sum_{\substack{0 \leq i \leq \ell \\ i \equiv j \mspace{-18mu} \pmod{h}}} \mspace{-9mu} \sum_{u \in E_i} m(u)
		\enspace.
	\]
\end{lemma}
\begin{proof}
Since $T_j$ is a subset of $E$, if we denote by $v_1, v_2, \cdots, v_m$ the element of $T_j$ in the order their arrival, then the submodularity of $f$ guarantees that
	\[ f(T_j \mid \varnothing) = \sum_{r=1}^{m} f(v_r \mid v_1, v_2, \cdots, v_{r-1}) \geq  \sum_{r=1}^{m} m(v_r) = \sum_{u \in T_j} m(u) \enspace. \]
Thus, to prove the lemma it suffice to prove
\begin{equation} \label{eq:target_T}
	\sum_{u \in T_j} m(u)
	\geq
	\frac{1}{4k} \cdot \mspace{-18mu} \sum_{\substack{0 \leq i \leq \ell \\ i \equiv j \mspace{-18mu} \pmod{h}}} \mspace{-9mu} \sum_{u \in E_i} m(u)
	\enspace.
\end{equation}

We prove Inequality~\eqref{eq:target_T} by proving a stronger claim via induction. However, before we can present this stronger claim, we need to define some additional notation. Recall that $T_j$ is constructed by starting with the empty set, greedily adding to it elements of $E_j$, then greedily adding to it elements of $E_{j + h}$, then greedily adding to it elements of $E_{j + 2h}$ and so on. Thus, let us define, for every integer $0 \leq i \leq \ell$ obeying $i \equiv j \pmod{h}$, the set $T^i_j$ to be the set $T_j$ immediately after \cref{alg:streaming-k-system} is done greedily adding elements of $E_i$ to $T_j$. Additionally, it is useful to define $T^{j - h}_j$ to be the empty set. Using these definitions, we can now define the stronger claim that we prove below by induction.

For every integer $-h \leq i \leq \ell$ and $j = i \bmod h$,
\begin{equation} \label{eq:induction_T}
	\sum_{u\in T_j^i} m(u) \geq \dfrac{1}{4k} \cdot \mspace{-18mu} \sum_{\substack{j \leq r \leq i \\ r \equiv j \mspace{-18mu} \pmod{h}}} \mspace{-9mu} \sum_{u \in E_r} m(u) + \dfrac{|T_j^i| \tau}{2^{i+2} k} \enspace.
\end{equation}
Before we prove this claim, let us observe that it indeed implies Inequality~\eqref{eq:target_T} by setting $i$ to be the largest integer that obeys $i \equiv j \pmod{h}$ and is not larger than $\ell$.

It now remains to prove Inequality~\eqref{eq:induction_T} by induction on $i$.
	For $i < 0$, this inequality holds since $j \geq 0 > i$ and $T_j^i = \emptyset$, which implies that the value of both sides of the inequality is $0$. Next, we need to prove Inequality~\eqref{eq:induction_T} for an integer $0 \leq i \leq \ell$ under the assumption that it holds for every $-h \leq i' < i$. Recall that the set $T_j^i$ is obtained by greedily adding elements of $E_i$ to $T_j^{i-h}$. Since the constraint is a $k$-system, the size of the set obtained in this way must be at least $|E_i|/k$ (otherwise, $T_j^i$ is a base of $E_i \cup T_j^{i - h}$ whose size is smaller than the size the independent set $E_i$ by more than a factor of $k$). Thus, we know that the number of elements of $E_i$ that are added to $T_j^{i - h}$ to form $T_j^i$ is at least $|E_i|/k - |T_j^{i-h}|$, which implies
	\begin{align*}
 \sum_{u \in T_j^{i} \backslash T_j^{i-h}} m(u) & \geq\left[\frac{\left|E_{i}\right|}{k}
 -\left|T_j^{i-h}\right|\right] \cdot \frac{\tau}{2^{i+1}}\\ &
 =\frac{\sum_{u \in E_{i}} \tau / 2^{i+2}}{k}+\frac{\left|E_{i}\right| \tau}{2^{i+2} k} -
 \frac{\left|T_j^{i-h}\right| \tau}{2^{i+1}}  \geq \frac{\sum_{u \in E_{i}} m(u)}{4 k}
 +\frac{\left|E_{i}\right| \tau}{2^{i+2} k}
 -\frac{\left|T_j^{i-h}\right| \tau}{2^{i+1}} \enspace,
	\end{align*}
	where the two inequality hold since $\tau/2^{i} \geq m(u) \geq \tau/2^{i+1}$ for every element $u \in E_i$.
	
Adding the induction hypothesis for $i - h$ to the above inequality, we get
\begin{align*}
 \sum_{u \in T_j^{i}} m(u) & \geq \frac{\sum_{u \in E_{i}} m(u)}{4 k}+\frac{\left|E_{i}\right| \tau}{2^{i+2} k}-
 \frac{\left|T_j^{i-h}\right| \tau}{2^{i+1}}+\dfrac{1}{4k} \cdot \mspace{-18mu} \sum_{\substack{j \leq r \leq i - h \\ r \equiv j \mspace{-18mu} \pmod{h}}} \mspace{-9mu} \sum_{u \in E_r} m(u) + \dfrac{|T_j^{i - h}| \tau}{2^{i-h+2} k} \\ 
 &=\dfrac{1}{4k} \cdot \mspace{-18mu} \sum_{\substack{j \leq r \leq i \\ r \equiv j \mspace{-18mu} \pmod{h}}} \mspace{-9mu} \sum_{u \in E_r} m(u)+\frac{\left|E_{i}\right| \tau}{2^{i+2} k}+\frac{\left|T_j^{i-h}\right| \tau}{2^{i+2} k}\left(2^{h}-2 k\right) 
 \\ & \geq \dfrac{1}{4k} \cdot \mspace{-18mu} \sum_{\substack{j \leq r \leq i \\ r \equiv j \mspace{-18mu} \pmod{h}}} \mspace{-9mu} \sum_{u \in E_r} m(u)+\frac{\left|E_{i}\right| \tau}{2^{i+2} k}+\frac{\left|T_j^{i-h}\right| \tau}{2^{i+2} k} \geq \dfrac{1}{4k} \cdot \mspace{-18mu} \sum_{\substack{j \leq r \leq i \\ r \equiv j \mspace{-18mu} \pmod{h}}} \mspace{-9mu} \sum_{u \in E_r} m(u)+\frac{\left|T_j^{i}\right| \tau}{2^{i+2 k}} \enspace,
 \end{align*}
 where the second inequality holds by the definition of $h$, and the last inequality holds since every element of $T_j^i$ must belong either to $E_i$ or to $T_j^{i-h}$.
\end{proof}

\begin{corollary} \label{cor:bound-T-E-tau-k-set}
	If $(\cN, \cI)$ is a $k$-set system, then \cref{alg:streaming-k-system} returns a set $T$ such that
	\[ 	f(T) \geq \dfrac{f(E \cup U) - \tau/4}{4kh(2k+1)} = \dfrac{f(E \cup U) - \tau/4}{O(k^2\log k)}\]
	for every set $U \in \cI$.
\end{corollary}
\begin{proof}
Since the output set of \cref{alg:streaming-k-system} is the best set among $T_0, T_1, \dotsc, T_{h - 1}$, we get
	\begin{align*}
		f(T) = \max_{0 \leq j < h} f(T_j) & \geq \dfrac{\sum_{j=0}^{h - 1} f(T_j)}{h} 
		\geq \dfrac{f(\varnothing) + \sum_{j = 0}^{h - 1} \sum_{i \in \{0 \leq i \leq \ell \mid i \equiv j \mspace{-18mu} \pmod{h}\}} \sum_{u \in E_i} m(u)}{4kh}  \\
		& = \dfrac{f(\varnothing) + \sum_{i=0}^{\ell} \sum_{u \in E_i} m(u)}{4kh} = \dfrac{f(\varnothing) + f(E \mid \varnothing)}{4kh} \geq \dfrac{f(E \cup U) - \tau/4}{4kh(2k+1)} \enspace,
	\end{align*}
	where the second inequality follows from Lemma~\ref{lem:T_system}, and the last inequality follows from Lemma~\ref{lem:E_bound} and the non-negativity of $f$.
\end{proof}

Using the last corollary and the framework described in \cref{sec:framework}, we can now prove our result for $k$-system constraints.
\begin{theorem} \label{thm:system}
There is a streaming $O(k^2 \log k) = \tilde{O}(k^2)$-approximation algorithm for the problem of maximizing a non-negative submodular function subject to a $k$-set system constraint.
\end{theorem}
\begin{proof}
If the objective function is monotone, then the theorem follows immediately from Corollary~\ref{cor:bound-T-E-tau-k-set} by setting $U$ to be the optimal solution, since this choice implies that the output set of \cref{alg:streaming-k-system} has a value of at least 
\[
	\dfrac{f(E \cup U) - \tau/4}{O(k^2\log k)}
	\geq
	\dfrac{\Opt - \Opt/2}{O(k^2\log k)}
	=
	\frac{\Opt}{O(k^2\log k)}
	\enspace,
\]
where the inequality holds since $\tau \leq 2M = 2\max_{u \in \cN, \{u\} \in \cI} f(\{u\}) \leq 2\Opt$ because $\{u\}$ is a candidate set to be $\Opt$ whenever it is feasible.

Otherwise, if the objective function is non-monotone, then we observe that Corollary~\ref{cor:bound-T-E-tau-k-set} implies that \cref{alg:streaming-k-system} is an $(O(k^2 \log k), \tau / 4)$-approximation algorithm when we take $A = E$. 
Thus, by setting $r = 4$, using \cref{alg:streaming-k-system} as \AlgStream and using the $(k + O(\sqrt{k}))$-approximation algorithm \textsc{RepeatedGreedy} due to~\cite{feldman2017greed} (mentioned in \cref{sec:application}) as \AlgConstrained, we get via our framework a streaming algorithm whose output set is guaranteed to have a value of at least
\[
	\dfrac{(r-1) \cdot \Opt - r\gamma}{r\alpha +  r(r-1)\beta/2}
	=
	\dfrac{3 \cdot \Opt - \tau}{O(k^2 \log k) +  6(k + O(\sqrt{k}))}
	=
	\dfrac{3 \cdot \Opt - \tau}{O(k^2 \log k)}
	\geq
	\frac{\Opt}{O(k^2 \log k)}
	\enspace.
	\qedhere
\]
\end{proof}

As promised, we now prove a stronger version of Lemma~\ref{lem:T_system} for $k$-extendible constraints. This version takes advantage of the stronger guarantee of the unweighted greedy algorithm for such constraints, which is given by Lemma~\ref{lem:greedy_improved_result}.

\begin{lemma} \label{lem:T_extendible}
	If $(\cN, \cI)$ is a $k$-extendible system, then for every integer $0 \leq j < h$ we have
	\[
		f(T_j \mid \varnothing) \geq \frac{1}{k} \cdot \mspace{-18mu} \sum_{\substack{0 \leq i \leq \ell \\ i \equiv j \mspace{-18mu} \pmod{h}}} \mspace{-9mu} \sum_{u \in E_i} m(u)
		\enspace.
	\]
\end{lemma}
	\begin{proof}
	We use in this lemma the notation defined in the proof of Lemma~\ref{lem:T_system}. Furthermore, the same arguments used in the proof of Lemma~\ref{lem:T_system} to show that Lemma~\ref{lem:T_system} follows from Inequality~\eqref{eq:induction_T} can also be used to show that the current lemma follows from the following claim. For every integer $-h \leq i \leq \ell$ and $j = i \bmod h$,
	\begin{equation} \label{eq:induction_T_extension}
	\sum_{u\in T_j^i} m(u) \geq \dfrac{1}{4} \cdot \mspace{-18mu} \sum_{\substack{j \leq r \leq i \\ r \equiv j \mspace{-18mu} \pmod{h}}} \mspace{-9mu} \sum_{u \in E_r} m(u) + \dfrac{|T_j^i| \tau}{2^{i+2}}
	\quad
	\enspace.
\end{equation}
Thus, the rest of this proof is devoted to proving this claim by induction on $i$.
	
	For $i < 0$, Inequality~\eqref{eq:induction_T_extension} holds since $j \geq 0 > i$ and $T_j^i = \emptyset$, which implies that the value of both sides of the inequality is $0$. Next, we need to prove Inequality~\eqref{eq:induction_T_extension} for an integer $0 \leq i \leq \ell$ under the assumption that it holds for every $-h \leq i' < i$.	Recall that the set $T_j^i$ is obtained by starting with $T_j^{i-h}$, and then greedily adding elements of $E_i$ to it. Thus, $T_j^i$ can be viewed as the output of the greedy algorithm when this algorithm is given the elements of $T_j^{i-h}$ first, and then the elements of $E_i$. Given this point of view, since $E_i$ is independent, Lemma~\ref{lem:greedy_improved_result} guarantees
	\[  |E_i \setminus T_j^i| \leq k \cdot |T_j^i \setminus E_i | = k \cdot |T_j^{i-h}| \enspace .  \]
	Hence,
	\begin{align*} \sum_{u \in T_j^{i} \backslash T_j^{i-h}} \mspace{-18mu} m(u) &=\sum_{u \in E_{i} \cap T_j^{i}} \mspace{-9mu} m(u) \geq\left|E_{i} \cap T_j^{i}\right| \cdot \frac{\tau}{2^{i+1}}=\left[\left|E_{i}\right|-\left|E_{i} \setminus T_j^{i}\right|\right] \cdot \frac{\tau}{2^{i+1}} \\ & \geq \sum_{u \in E_{i}} \frac{\tau}{2^{i+2}}+\frac{\left|E_{i}\right| \cdot \tau}{2^{i+2}}-\frac{k \tau\left|T_j^{i-h}\right|}{2^{i+1}} \geq \frac{\sum_{u \in E_{i}} m(u)}{4}+\frac{\left|E_{i}\right| \cdot \tau}{2^{i+2}}-\frac{k \tau \cdot\left|T_j^{i-h}\right|}{2^{i+1}} \enspace,
	 \end{align*}
	 where the first and third inequalities hold since $\tau/2^i \geq m(u) \geq \tau/2^{i+1}$ for every element $u \in E_i$.
	 
Adding the induction hypothesis for $i - h$ to the above inequality, we get
	\begin{align*} 
		\sum_{u \in T_j^{i}} m(u) & \geq \frac{\sum_{u \in E_{i}} m(u)}{4}+\frac{\left|E_{i}\right| \cdot \tau}{2^{i+2}}-\frac{k \tau \cdot\left|T_j^{i-h}\right|}{2^{i+1}}+\dfrac{1}{4} \cdot \mspace{-18mu} \sum_{\substack{j \leq r \leq i - h \\ r \equiv j \mspace{-18mu} \pmod{h}}} \mspace{-9mu} \sum_{u \in E_r} m(u)+\frac{\left|T_j^{i-h}\right| \tau}{2^{i-h+2}} \\ 
		&
		=\dfrac{1}{4} \cdot \mspace{-18mu} \sum_{\substack{j \leq r \leq i \\ r \equiv j \mspace{-18mu} \pmod{h}}} \mspace{-9mu} \sum_{u \in E_r} m(u)+
		\frac{\left|E_{i}\right| \cdot \tau}{2^{i+2}}+\frac{\left|T_j^{i-h}\right| \cdot \tau}{2^{i+2} } \left(2^{h}-2 k\right) \\ 
		& \geq \dfrac{1}{4} \cdot \mspace{-18mu} \sum_{\substack{j \leq r \leq i \\ r \equiv j \mspace{-18mu} \pmod{h}}} \mspace{-9mu} \sum_{u \in E_r} m(u)+\frac{\left|E_{i}\right| \cdot \tau}{2^{i+2}}+\frac{\left|T_{i-h}\right| \cdot \tau}{2^{i+2}} \geq \dfrac{1}{4} \cdot \mspace{-18mu} \sum_{\substack{j \leq r \leq i \\ r \equiv j \mspace{-18mu} \pmod{h}}} \mspace{-9mu} \sum_{u \in E_r} m(u)+\frac{\left|T_{i}\right| \cdot \tau}{2^{i+2}} \enspace,
	\end{align*}
where the second inequality holds by the definition of $h$, and the last inequality holds since every element of $T_j^i$ belongs either to $E_i$ or to $T_j^{i-h}$.	
	\end{proof}
	
	Using Lemma~\ref{lem:T_extendible} we can prove the following theorem. We omit the proof of this theorem since it is identical to the proof of \cref{thm:system}, except for the use Lemma~\ref{lem:T_extendible} instead of Lemma~\ref{lem:T_system}.
	
\begin{theorem} \label{thm:extendible}
There is a streaming $O(k \log k) = \tilde{O}(k)$-approximation algorithm for the problem of maximizing a non-negative submodular function subject to a $k$-extendible system constraint.
\end{theorem}

\section{Experiments} \label{sec:experiments}

In this section, we compare our proposed algorithms with two other groups of algorithms: other streaming algorithms and state-of-the-art \emph{offline} algorithms. For the streaming algorithms we consider three algorithms: 
1) streaming-greedy: this algorithm keeps a solution $S$ which is initially set to the empty set. For every incoming element $u$, it is added to the set $S$ if this does not violate feasibility (i.e., $S \cup \{u\} \in \cI$). 
2)  Inspired by the streaming algorithms of \cite{chekuri2015streaming, buchbinder2019online}, we consider a heuristic preemptive algorithm. 
For every incoming element $u$, this algorithm generates a set $U \subseteq S$ such that $(S \cup \{u\}) \setminus U$ is feasible under the non-knapsack constraints. The element $u$ is then added to the solution in exchange for the elements of $U$ if this does not violate the knapsack constraints and the exchange is beneficial in the sense that $f(u \mid S) \geq  \sum_{u' \in U} f(u':S)$, where $f(u':S) = f(u' \mid S')$ for $S' = \{s \in S : \textrm{element } s \textrm{ arrived before } u'\}$. For more detail refer to \cite{chekuri2015streaming, feldman2018doless}. 
3) sieve-streaming: this heuristic algorithm is implemented based on the ideas of \cite{badanidiyuru2014streaming}. 
In the first step, it finds an accurate estimation of $\Opt$. Then, each incoming element $u$ is added to the solution $S$ if $S \cup \{ u\} \in \cI$ and $f(u \mid S) \geq \Opt / (2\maxcardinality)$, where $\maxcardinality$ is the maximum cardinality of a feasible solution. 
For the offline algorithms we consider 1) the vanilla greedy algorithm, 2) Fast \cite{badanidiyuru2014fast} and FANTOM \cite{mirzasoleiman2016fast}.
Both the Fast and FANTOM algorithms are designed to maximize submodular functions under a $p$-system constraint combined with $\ell$ knapsack constraints.

\cref{sec:independent-set,sec:planarity} compare the above algorithms on tasks of maximizing linear and cut objective functions over instances produced using synthetic and real-world data, respectively. Then, in \cref{sec:movie,sec:yelp,sec:twitter}, we evaluate the performance of the same algorithms on three different real-world applications.
In a movie recommendation system application, we are given movie ratings from users, and our goal is to recommend  diverse movies from different genres.
In a Yelp location summarization application, we are given thousands of business locations with several related attributes. Our objective is to find a good summary of the locations from the following cities: Charlotte, Edinburgh, Las Vegas, Madison, Phoenix and  Pittsburgh.
In a third application, our goal is to generate real-time summaries for Twitter feeds of several news agencies with the following Twitter accounts (also known as ``handles''): @CNNBrk, @BBCSport, @WSJ, @BuzzfeedNews, @nytimes, @espn.

\subsection{Independent set} \label{sec:independent-set}

In the experiments of this section we define submodular functions over the nodes of a given graph $G= (V, E)$, and consider the maximization of such functions subject to an independent set constraint, i.e., we are not allowed to select a set of vertices if there is any edge of the graph connecting any two of these vertices. It is easy to show that this constraint is a $d_{\max}$-extendible system, where $d_{\max}$ is the maximum degree in graph $G$. In our experiments in this section, we use two types of synthetically generated random graphs: Erd\H{o}s R\'eny graphs  \cite{erdos1960on} and Watts–Strogatz  graphs \cite{watts1998collective}. For the Erd\H{o}s R\'eny graphs we vary in our experiments the probability $p$ that each possible edge is included in the graph (independently from every other edge), and for the Watts-Strogatz graphs we vary the rewiring probability $\beta$. The number of nodes is set to $n=2000$ in all the graphs, and in the Watts–Strogatz model each node is connected to $k=100$ nearest neighbors in the ring structure. For more detail regarding these random graph models refer to \cite{hofstad2016random}.

In our first experiment of the section, we study the maximization of the following monotone linear function
\begin{align} \label{eq:linear-graph}
f(S) = \sum_{u \in S} w_u \enspace,
\end{align}
where $S \subseteq V$ and $w_u$ is the weight of node $u \in S$. \cref{fig:f-gnp-sw,fig:f-ws-sw} compare different algorithms for optimizing this function over a random graph chosen from the above discussed random graph models. 
We observe that our proposed algorithm consistently outperforms the other baseline streaming algorithms. We also observe that the performance our algorithm is comparable with (or even better at times) the greedy algorithm, which is provably optimal for the maximization of linear functions subject to $k$-extendible constraints~\cite{feldman2017greed}.

In the second experiment, we consider the non-monotone submodular graph-cut function:
\begin{align} \label{eq:cut-graph}
	f(S) = \sum_{u \in S} \sum_{v \in V \setminus S} w_{u,v} \enspace,
\end{align}
where $w_{u,v} $ is the weight of the edge $e = (u,v)$.
Again, in \cref{fig:f-gnp-cut,fig:f-ws-cut} we observe that the solutions provided by our algorithms are clearly better than those produced by the other streaming algorithms. 
Furthermore, the non-monotone version of our algorithm always outperforms the monotone algorithm. We also note that the for this non-monotone submodular function, the vanilla greedy algorithm performs very poorly (which is consistent with the lack of a theoretical guarantee for this algorithm for such functions).

\begin{figure*}[ht] 
	\centering  
	\subfloat[Erd\H{o}s R\'eny (linear)] {\includegraphics[height=30.9mm]{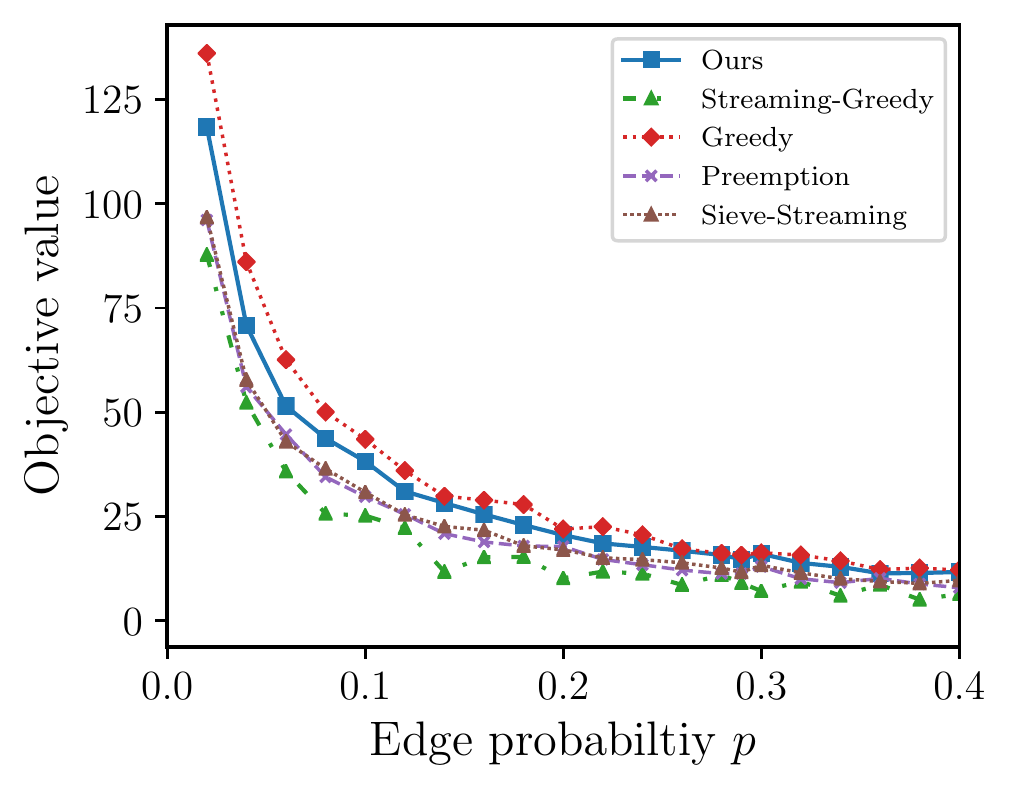}\label{fig:f-gnp-sw}} 
	\subfloat[Watts–Strogatz (linear)]	{\includegraphics[height=30.9mm]{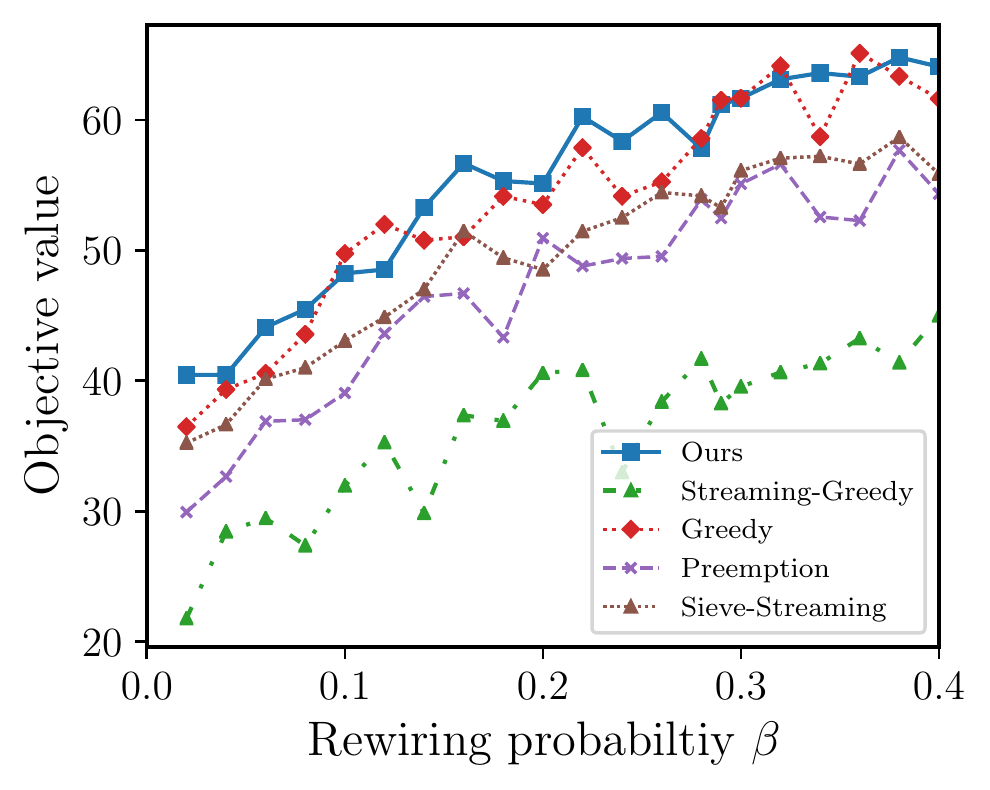}\label{fig:f-ws-sw}}
		\subfloat[Erd\H{o}s R\'eny (cut)]	{\includegraphics[height=30.9mm]{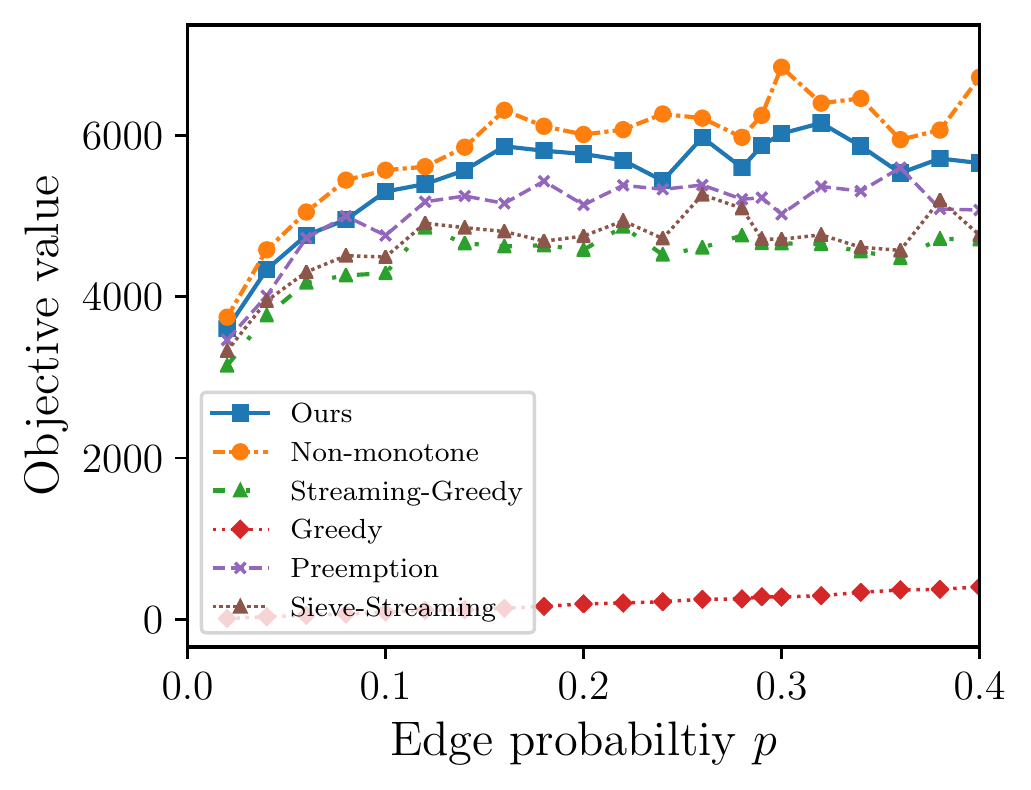}\label{fig:f-gnp-cut}}
	\subfloat[Watts–Strogatz (cut)]	{\includegraphics[height=30.9mm]{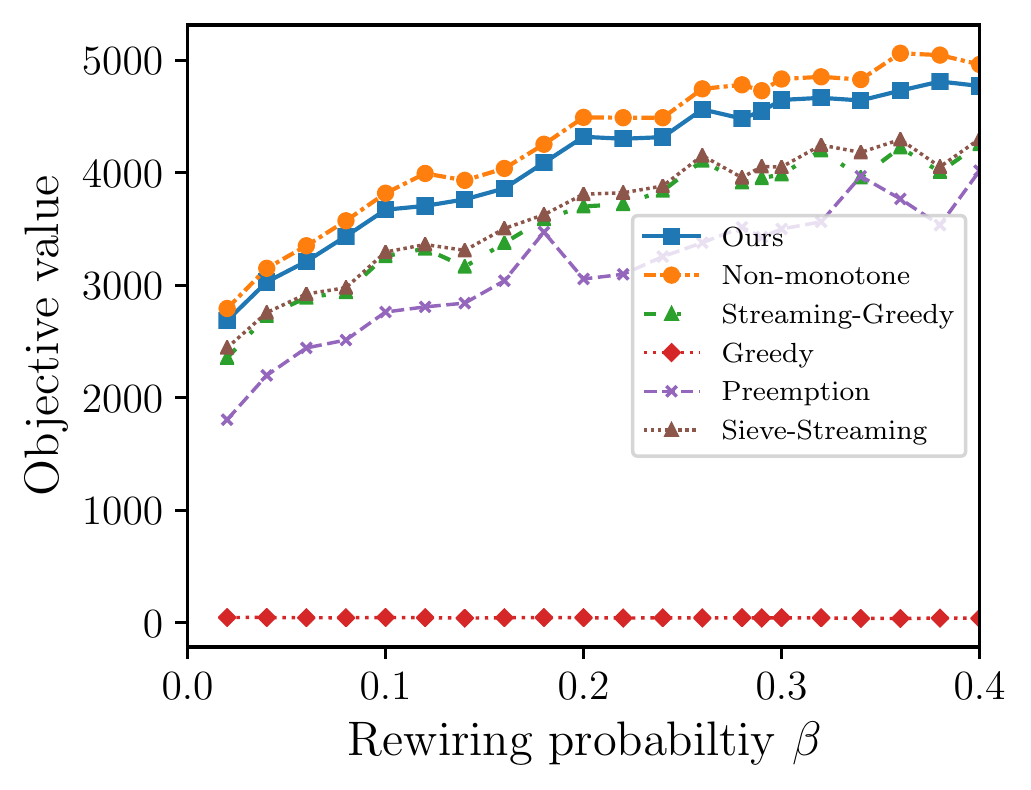}\label{fig:f-ws-cut}}
		\caption{For Erd\H{o}s R\'eny graphs $p$ is the probability of having an edge between any two nodes. For Watts–Strogatz graphs $\beta$ is the probability of rewiring of each edge.\label{fig:independet-set}}
\end{figure*}

\subsection{Graph Planarity with Knapsack} \label{sec:planarity}

In the experiment of this section, our objective is to maximize a linear function over the edges of a graph $G = (V,E)$. For the constraint, we require that an independent set of edges corresponds to a planar sub-graph of $G$, and in addition, it satisfies a given knapsack constraint. 
We remind the reader that a graph is planar if it can be embedded in the plane.
Furthermore, a knapsack constraint is defined by a cost function $c\colon \cN \to \bR_{\geq 0}$, and we say that a set $S \subseteq \cN$ satisfies the knapsack constraint if $c(S) = \sum_{e \in S} c(e) \leq b$ for a given knapsack budget $b$.

Before describing the experiments we did in more detail, let us explain why the above mentioned constraint is $k$-system for a (relatively) modest value of $k$. It is straightforward to show that a single knapsack constraint is a $\ceil{\nicefrac{c_{\max}}{c_{\min}}}$-extendible system, and consequently a  $\ceil{\nicefrac{c_{\max}}{c_{\min}}}$-set system, where $c_{\max} = \max_{e \in \cN} c(e)$ and $c_{\min} = \min_{e \in \cN} c(e)$. We complement this with Lemmata~\ref{lemma:planar} and~\ref{lemma:two-k-systems}, which prove that the planarity constraint is a $3$-set system and that the intersection of a $k_1$-set system and a $k_2$-set system is a $(k_1 + k_2)$-set system, respectively. We would like to thank Chandra Chekuri for pointing Lemma~\ref{lemma:planar} to us. We also would like to state that Lemma~\ref{lemma:two-k-systems} is similar to well-known properties of more restricted classes of set systems, but we are not aware of a previously published explicit proof of it for general $k$-set systems.

\begin{lemma} \label{lemma:planar}
	For every graph $G = (V, E)$, the system $\cM = (E, \cI)$, where $\cI= \{S \subseteq E \mid (V,S) \text{ is a planar}\allowbreak \text{ graph} \}$, is a $3$-set system.
\end{lemma}
\begin{proof}
	Note that $\cM$ is downward closed, because a subgraph of a planar graph is also planar. Furthermore, the empty set is always a member of $\cI$ because a graph with no edges is planar. Thus, we concentrate on proving that $\cM$ obeys the remaining property of $k$-set systems. Formally, for an arbitrary set $E' \subseteq E$, and two arbitrary bases $B_1$ and $B_2$ of $E'$ (i.e., subsets of $E'$ which are independent, and no other edge of $E'$ can be added to them without violating independence), we need to show $|B_1| / |B_2| \leq 3$.
	
	Assume first, for simplicity, that $(V, E')$ is connected. This implies that $(V, B_2)$ is also connected (otherwise, we can add to it any edge connecting two different connected components without violating planarity, which contradicts the fact that it is a base of $E'$), and thus, the size of $B_2$ must be at least $|V| - 1$.
	Additionally, it is well-known that, using Euler's formula, it is possible to show that the number of edges in a planar graph is at most $3 |V| - 6$ as long as $|V| \geq 3$. Thus, $|B_1| \leq 3|V| - 6$ as long as $|V| \geq 3$. For $|V| < 3$, we still get $|B_1| \leq 3|V| - 3$ because a graph with a single vertex can include no edges and a graph with two vertices can include at most a single edge. Combining all these observations, we now get
	\[ \dfrac{|B_1|}{|B_2|} \leq \dfrac{3 |V| - 3}{|V| - 1} = 3 \enspace. \]   	
	
	Consider now the case in which $(V, E')$ has more than one connected component. In this case we can use the above argument for each component of $(V, E')$. Thus, if we denote by $m$ the number of connected components of this graph, where $V_i$ is the set of vertices of the $i$-th component, then we get
	\[ \dfrac{|B_1|}{|B_2|} \leq \dfrac{\sum_{i=1}^{m} 3 |V_i| - 3}{\sum_{i=1}^{m}|V_i| - 1} = 3 \enspace. \qedhere \]   
\end{proof}


\begin{lemma} \label{lemma:two-k-systems}
	Let $\cM_1 = (\cN, \cI_1)$ and $\cM_2 = (\cN, \cI_2)$ be a $k_1$-set system and a $k_2$-set system, respectively, over the same ground set $\cN$. Then, the set system $\cM = (\cN, \cI_1 \cap \cI_2)$ is a $(k_1 + k_2)$-set system.
\end{lemma}
\begin{proof}
	For every subset $F \subseteq \cN$, let $\cB_\cM(F)$ be the set of bases of $F$ with respect to $\cM$. It is clear that $\cI_1 \cap \cI_2$ is down-monotone and contains the empty set. Thus, to prove the lemma we only need to prove that for every subset $F \subseteq \cN$
	\begin{equation} \label{eq:ratio}
	\frac{\max_{B \in \cB_\cM(F)} |B|}{\min_{B \in \cB_\cM(F)} |B|}
	\leq
	k_1 + k_2
	\enspace.
	\end{equation}
	
	Towards this goal, let us define $B_\ell = \arg \max_{B \in \cB_\cM(F)} |B|$ and $B_s = \arg \min_{B \in \cB_\cM(F)} |B|$. For every $i \in \{1, 2\}$, let $D_i$ be the set of elements of $B_\ell$ that do not belong to $B_s$ and cannot be added to $B_s$ without violating independence with respect to $\cM_i$. Formally, $D_i = \{u \in B_\ell \setminus B_s \mid B_s + u \not \in \cI_i\}$. Since $B_s$ is a base of $\cM$, every element of $B_\ell \setminus B_s$ must belong either to $D_1$ or to $D_2$. Hence, we get
	\[
	|D_1| + |D_2|
	\geq
	|B_\ell \setminus B_s|
	\enspace.
	\]
	
	Observe now that for every $i \in \{1, 2\}$ the set $U_i = D_i \cup (B_\ell \cap B_s)$ is a subset of $B_\ell$, and thus, independent with respect to $\cM_i$. Moreover, the definition of $D_i$ implies that $B_s$ is a base of $U_i \cup B_s$ with respect to $\cM_i$. Since $\cM_i$ is a $p_i$-system, this implies that the size of the independent set $U_i$ is upper bounded by $p_i \cdot |B_s|$. Thus, we get
	\[
	|D_i| + |B_\ell \cap B_s|
	=
	|U_i|
	\leq
	p_i \cdot |B_s|
	\quad
	\forall\; i \in \{1, 2\}
	\enspace.
	\]
	
	Combining the above inequalities, we get
	\[
	|B_\ell|
	\leq
	|D_1| + |D_2| + |B_\ell \cap B_s|
	\leq
	p_1 \cdot |B_s| + p_2 \cdot |B_s| - |B_\ell \cap B_s|
	\leq
	(p_1 + p_2) \cdot |B_s|
	\enspace,
	\]
	which proves Inequality~\eqref{eq:ratio} due to the definitions of $B_s$ and $B_\ell$.
\end{proof}

We now get back to the experiment of this section. Recall that in this experiment the goal is to maximize a submodular function $f$ under the combination of a graph planarity constraint and a single knapsack constraint $c$. For the objective function $f$, we use the monotone linear function:
\[f(S) = \sum_{e \in S} w_e \quad \forall S \subseteq E \enspace, \] 
where $w_e$ is the weight of edge $e \in S$, and for simplicity, we set all these weights to $1$. 
The knapsack cost of each edge $e = (u, v) \in E$ is chosen to be proportional to $\max(1, d_u  - q)$, where $d_u$ is the degree of node $u$ in graph $G$ and $q=6$, and the costs are normalized so that $\sum_{e \in E} c_e = |V|$, where $c_e$ represents the knapsack cost of edge $e$.

In the experiment, we use four real-world networks from \cite{snapnets} as the graph, and vary the knapsack budget between $0$ and $1$ (note that the normalization gives this range of budgets an intuitive meaning). In \cref{fig:planarity_linear} we compare the performance of our streaming algorithm with the performance of Streaming Greedy and Sieve Streaming. One can observe that our algorithm outperforms the two other baselines. Due to the prohibitive computational complexity of the offline algorithms, we do not report their results for this experiment. Furthermore, as it is not clear how to execute a preemptive streaming algorithm under a planarity constraint, we did not include a version of the Preemption algorithm in this experiment.

\begin{figure*}[ht] 
	\centering  
	\subfloat[Social graph]	{\includegraphics[height=32mm]{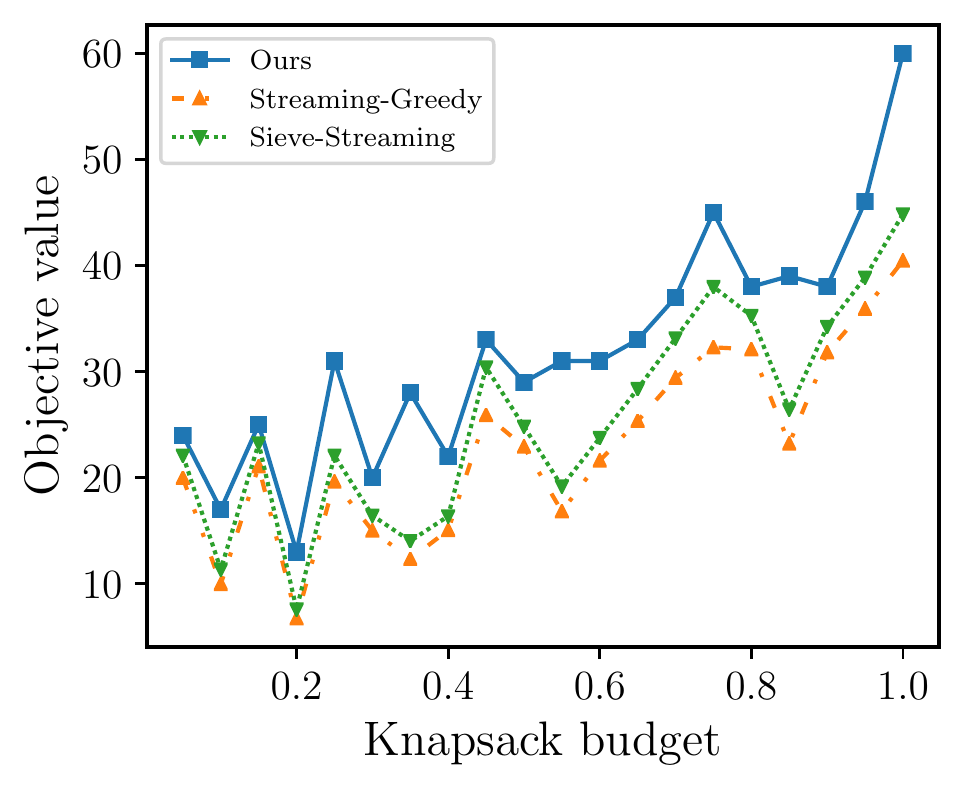}\label{fig:planar-0}}
	\subfloat[EU Email] {\includegraphics[height=32mm]{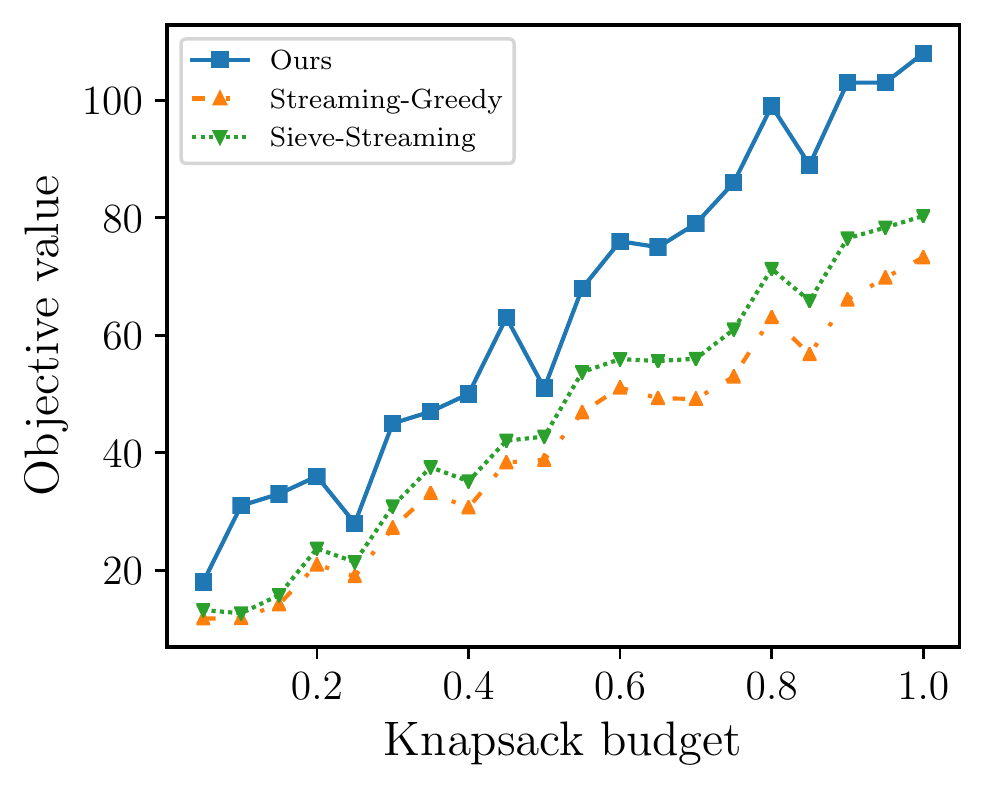}\label{fig:planar-EU}} 
	\subfloat[Facebook ego network]	{\includegraphics[height=32mm]{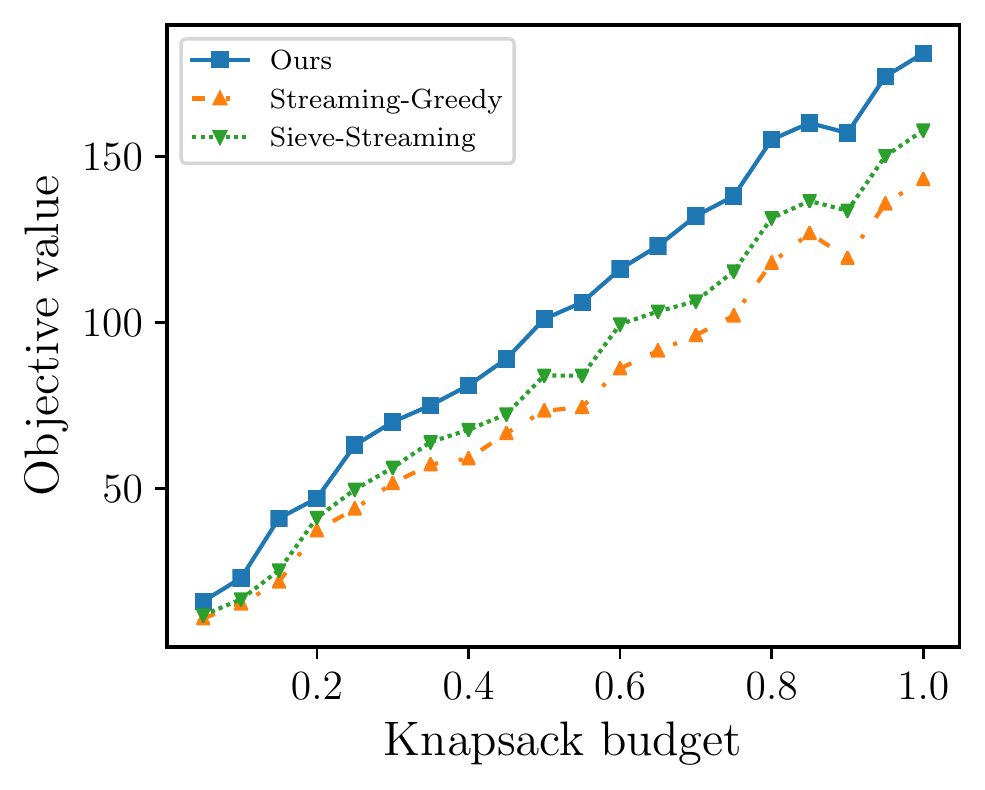}\label{fig:planar-1912}}
	\subfloat[Wiki vote network]	{\includegraphics[height=32mm]{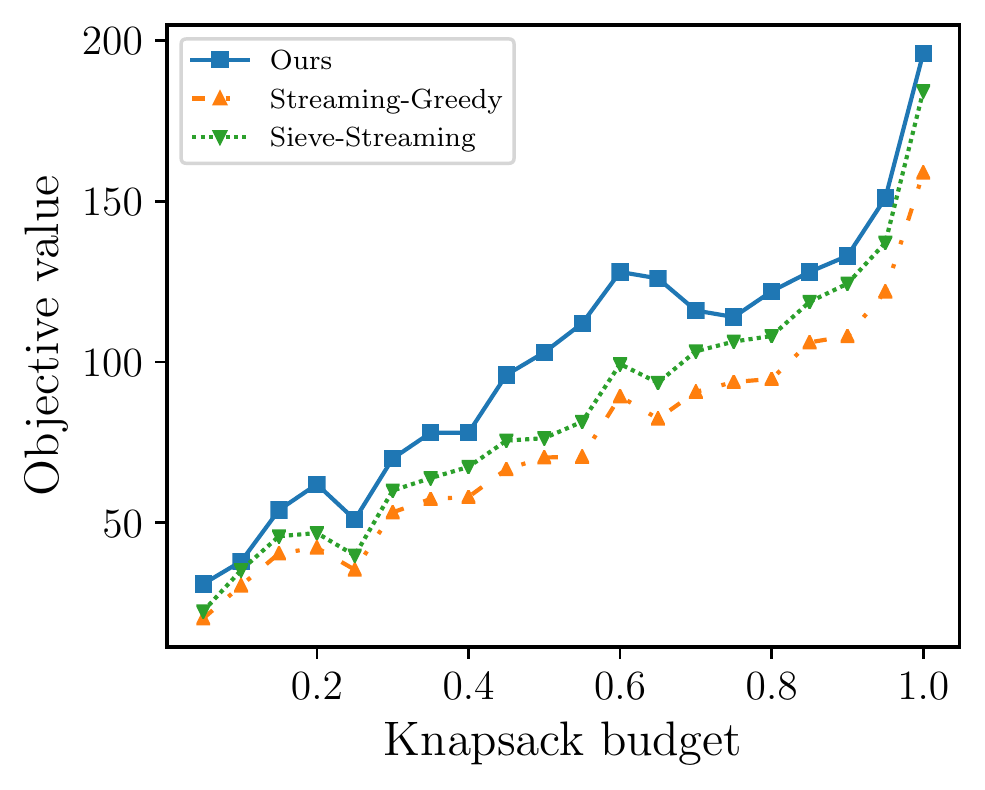}\label{fig:planar-wiki}}
	\caption{Planarity with knapsack (linear objective function).
		The weight of each edge is set to one.
		Knapsack cost of each edge $e = (u, v)$ is proportional to $\max(1, d_u  - 6)$, where $d_u$ is the degree of node $u$ in graph $G$. The costs are normalized so that $\sum_{e \in E} c_e = |V|$, where $c_e$ represents the knapsack cost of edge $e$.
	\label{fig:planarity_linear}
}
\end{figure*}


In \cref{fig:planarity_linear_new} (appearing in \cref{sec:sup_experiments}), we consider another knapsack constraint. In this constraint the cost of each edge $e = (u, v)$ is proportional to an integer picked uniformly at random from the set $\{1,2,3,4,5 \}$, and the costs are normalized as in the previous experiment. 
Again, we observe that our streaming algorithms returns solutions with higher objective values for various knapsack budgets.

\subsection{Movie Recommendation} \label{sec:movie}

In the movie recommendation application, our goal is to select a diverse set of movies subject to constraints that can be adjusted by the user. The dataset for this experiment contains 1793 movies from the genres: Adventure, Animation and Fantasy (note that a single movie may be identified with multiple genres). The user may specify an upper limit $m$ on the number of movies in the set we recommend for them, as well as an upper limit $m_i$ on the number of movies from each genre. For simplicity, we use a single value for all $m_i$ and refer to this value as the genre limit. It is easy to show that this set of constraints forms a $3$-extendible system. In addition, we enforce two knapsack constraints. For the first knapsack constraint $c_1$, the cost of each movie is proportional to the absolute difference between the release year of the movie and the year $1985$ (the implicit goal of this constraint is to pick movies with a release year which is as close as possible to the year $1985$). 
For the second knapsack constraint $c_2$, the cost of each movie is proportional to the difference between the maximum possible rating (which is $10$) and the rating of the particular movie---here the goal is to pick movies with higher ratings. 
More formally, for a movie $v \in \cN$, we have: $c_1(v) \propto \lvert 1985- \textrm{year}_v \rvert$ and  $c_2(v) \propto ( 10 - \textrm{rating}_v )$.
Here, $\textrm{year}_v$ and $\textrm{rating}_v$, respectively, denote the release year and  IMDb rating of movie $v$.
We normalize the costs in both knapsacks constraints so that the average cost of each movie is $\nicefrac{1}{10}$, i.e., $\frac{\sum_{v \in \cN}c_i(V)}{|\cN|} = \nicefrac{1}{10}$, and we set the knapsack budgets to $1$. Intuitively, this choice means that we expect a feasible set to contain no more than roughly $10\%$ of the movies.

In our experiments, we try to maximize two kinds of objective functions (each trying to capture diversity in a different way) subject to these constraints, and we vary the upper limit $m$ on the number of movies in the recommended set of movies. Both objective functions are based a set of attributes calculated for each movie using the method described in~\cite{lindgren2015sparse}, and both objective functions are non-negative and submodular. However, one of them is monotone, and the other is not (guaranteed to be) monotone.

\subsubsection{Monotone Submodular Function} \label{sec:movie-monotone}

In this section we describe the part of the experiment using a non-negative, monotone and submodular objective function. Let us begin the section by describing this function. Assume $v_i$ represents the feature vector of the $i$-th movie, then we define a matrix $M$ such that $M_{ij} = e^{-\lambda \cdot  \text{dist}(v_i,v_j)}$, where $\text{dist}(v_i,v_j)$ is the euclidean distance between vectors $v_i, v_j$---informally $M_{ij}$ encodes the similarity between the frames represented by $v_i$ and $v_j$.
The diversity of a set $S$ of movies is measured by the non-negative monotone submodular objective $f(S) = \log \det(\mathbf{I} + \alpha M_S )$, where $\mathbf{I}$ is the identity matrix, $\alpha$ is a positive scalar and $M_S$ is the principal sub-matrix of $M$ indexed by $S$ \citep{herbrich2003fast}.

In the experiment we did with the above objective function, we set the genre limit to 10, $\lambda$ to $0.1$, and $\alpha$ to $20$. The results of the experiment appear in \cref{fig:movie-f-m,fig:movie-o-m}. In \cref{fig:movie-f-m}, we can observe that our streaming algorithm outperforms streaming greedy and sieve streaming. Furthermore, while our algorithm requires only a single pass over the data and enjoys a very low computational complexity (see \cref{fig:movie-o-m}), the solutions it returns are competitive with respect to the solutions produced by the offline algorithms we compare with.

\subsubsection{Non-monotone Submodular Function}

An intuitive utility function for choosing a diverse set of movies $S$ is the following not necessarily monotone submodular function
\begin{equation} \label{eq:cut_fun}
f(S) = \sum \limits_{i \in S} \sum \limits_{j \in \cN} M_{i,j} - \sum \limits_{i \in S} \sum \limits_{j \in S} M_{i,j} \enspace,
\end{equation}
where $\cN$ is the set of all movies and $M_{i,j}$ is the non-negative similarity score between movies $i,j \in \cN$ as defined in the previous section.
It is beneficial to note that the first term is a sum-coverage function that 
captures the representativeness of the selected set, and the second term is a dispersion function penalizing 
similarity within $S$ \cite{feldman2017greed}.

In our experiment with this function as the object, we set the genre limit to 20. 
In \cref{fig:movie-f-nm,fig:movie-o-nm},  we observe that i) our streaming algorithm returns solutions with higher utilities in comparison to the baseline streaming algorithms, ii) the non-monotone version of our streaming algorithm clearly outperforms the monotone one for this non-monotone submodular function, and iii) the quality of the solutions returned by our algorithms is comparable with the quality obtained by offline algorithms.

\begin{figure*}[ht] 
	\centering  
	\subfloat[Monotone function] {\includegraphics[height=32mm]{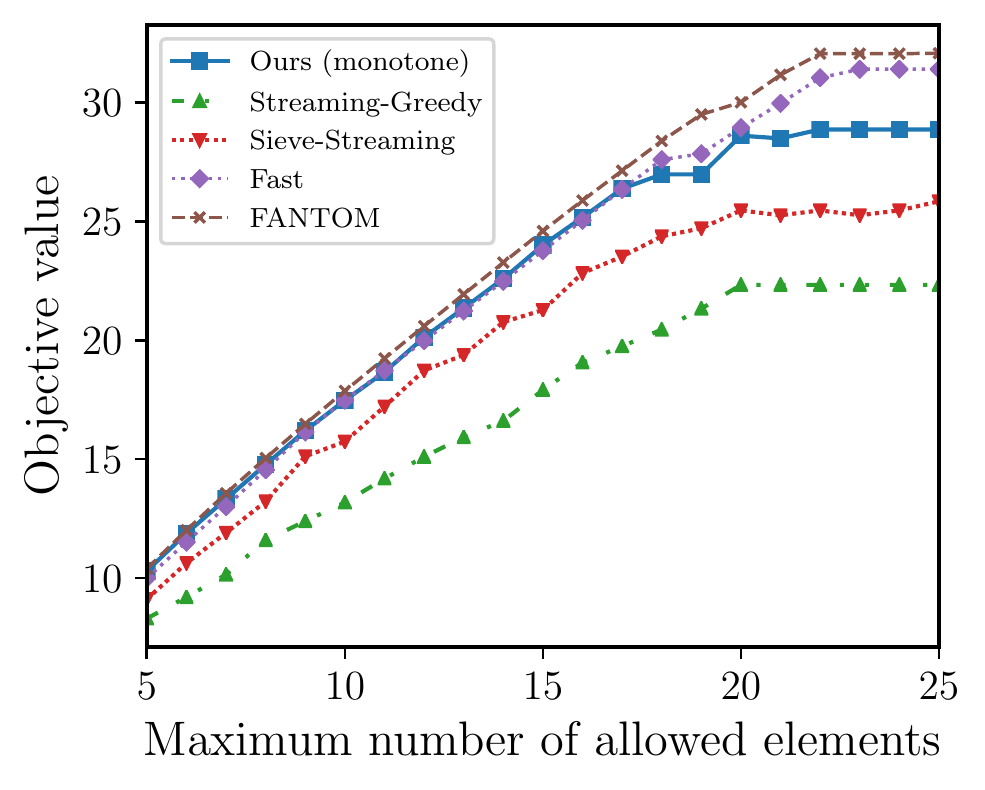}\label{fig:movie-f-m}} 
	\subfloat[Monotone function]	{\includegraphics[height=32mm]{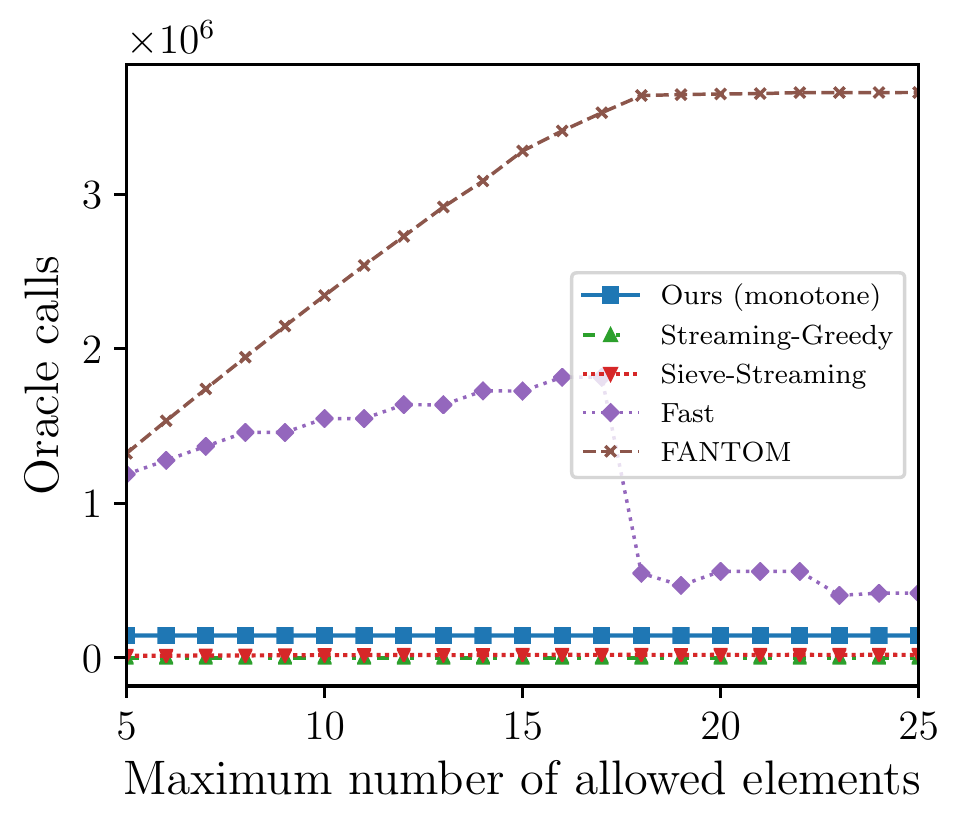}\label{fig:movie-o-m}}
	\subfloat[Non-monotone function]	{\includegraphics[height=32mm]{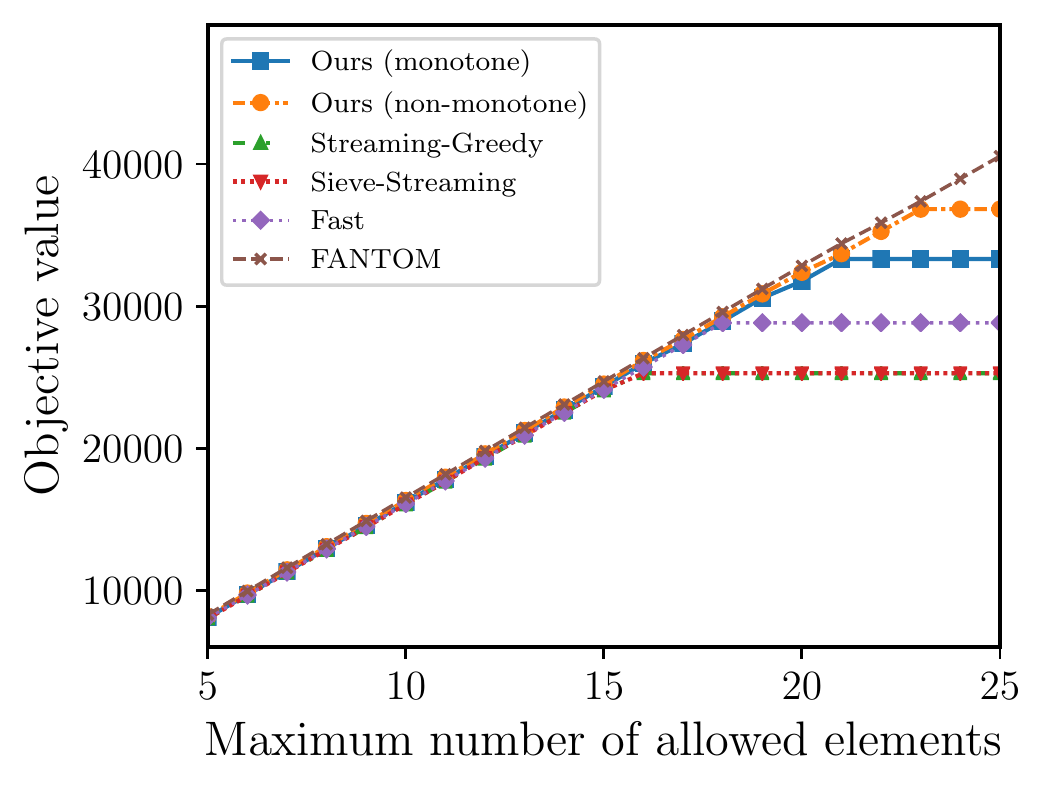}\label{fig:movie-f-nm}}
	\subfloat[Non-monotone function]	{\includegraphics[height=32mm]{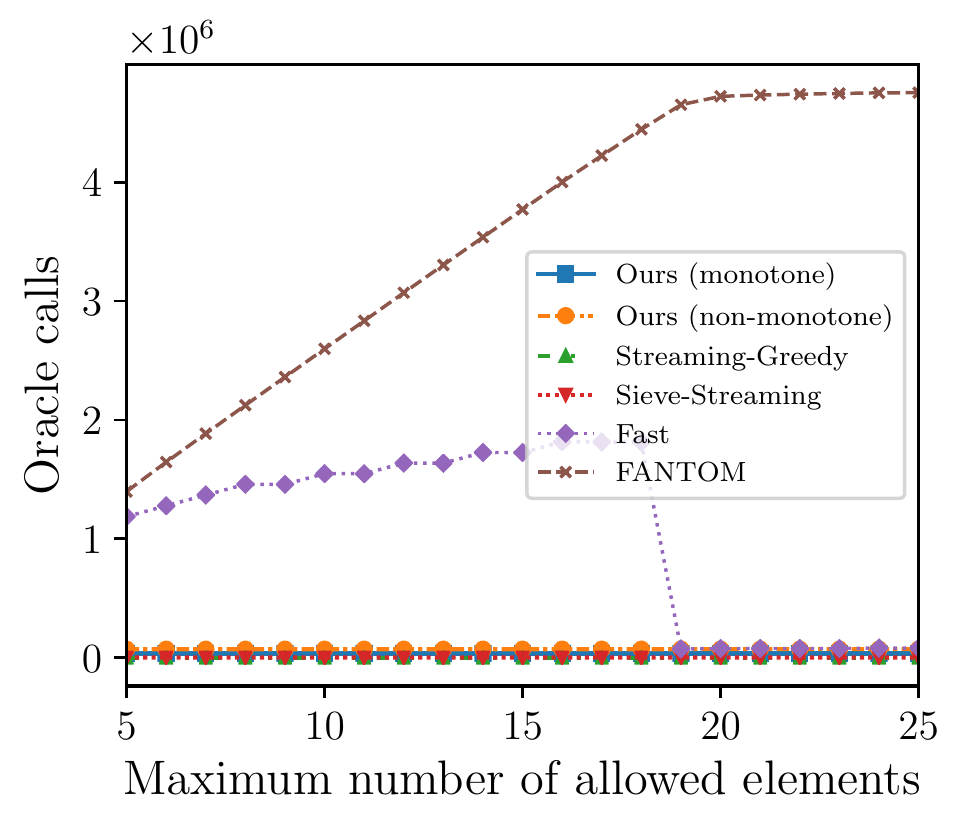}\label{fig:movie-o-nm}}
	\label{fig:movie}
	\caption{Movie recommendation with two knapsacks}
\end{figure*}

\subsection{Yelp Location Data Summarization} \label{sec:yelp}

In this application, our objective is to summarize a large dataset of locations. 
We use the Yelp Academic dataset~\cite{yelp}, which is a subset of Yelp's businesses, reviews, and user data~\cite{yelporig}. 
The dataset contains information about local businesses across 11 metropolitan areas, and we consider only locations in six out of these metropolitan areas.
We used the description of each business location and reviews for feature extraction. 
These features contain information regarding many attributes such as having vegan menus, delivery options,  possibility of outdoor seating,  being good for groups, etc.\footnote{For the feature extraction, we used the script provided at \url{https://github.com/vc1492a/Yelp-Challenge-Dataset}.}

Suppose we want to select, out of a ground set $\ground = \{1, \dots , n \} $, a subset of locations
which provides a good representation of all the existing business locations. Towards this goal, we calculate a matrix $M$ representing the similarity between every two locations $i,j \in \cN$ using the same method described in \cref{sec:movie-monotone}. Then, intuitively, given a set $S$, each location $i \in \cN$ is represented by the location from the set $S$ with the highest similarity to $i$. Thus, it is natural to define the total utility provided by a set $S$ using the following non-negative, monotone and submodular set function \cite{krause12survey,frieze1974cost}:
\begin{equation} \label{eq:facility}
f(S) =\frac{1}{n} \sum_{i=1}^{n} \max_{j \in S} M_{i,j} \enspace.
\end{equation}
Note that the utility function \eqref{eq:facility} depends on the entire dataset $\cN$.
In the streaming setting we do not have access to the full data stream, but fortunately, our objective function is additively decomposable~\cite{mirzasoleiman2013distributed} over the ground set $\cN$. 
Thus, as long as we can sample uniformly at random from a data stream, it is possible to estimate \eqref{eq:facility} arbitrarily close to its exact value \cite[Proposition 6.1]{badanidiyuru2014streaming}.
To sample randomly from the data stream and estimate the function, we use the reservoir sampling technique explained in \cite[Algortithm~4]{badanidiyuru2014streaming}.

For the constraint, we use a combination of matroid and knapsack constraints (which yields a $\ksys$-extendible constraint).
The matroid constraint is as follows: i) there is a limit $m$ on the total number of selected locations and ii) the maximum number of allowed locations from each of the six cities is 10.
For the knapsack constraints we consider two different scenarios: 
i) in the first scenario, there is a single knapsack $c_1$ in which the cost assigned to each location is proportional to the distance of that location from a pre-specified location in the down-town of its metropolitan area. 
ii) in the second scenario, we add another knapsack $c_2$ which is based on the distance between each location and the international airport serving its metropolitan area. 
In this set of experiments, we set the knapsack budgets to $1$, where one unit of budget is equivalent to 100km. 
This means that we allow the sum of the distances of every feasible set of locations to the points of interest (i.e., down-towns or airports) to be at most 100km.

In our experiments, we compare the utility and computational cost of algorithms for different values of $m$ (the upper limit on the number of locations in the produced summary). 
From the experiments (see \cref{fig:yelp}), we observe that i) our proposed algorithm, consistently, demonstrates a better performance compared to other streaming algorithms in terms of the utility of the final solution, and ii) the utility of the solutions produced by our algorithm is comparable to the utility of solutions produced by state-of-the-art offline algorithms, despite the ability of our algorithm to make only a single pass over the data and its several orders of magnitude better computational complexity. We also observe that, as expected, adding more constraints (compare \cref{fig:yelp-f-one,fig:yelp-f-two}) reduces the utility of the selected summary.

\begin{figure*}[ht] 
	\centering  
	\subfloat[One knapsack] {\includegraphics[height=32mm]{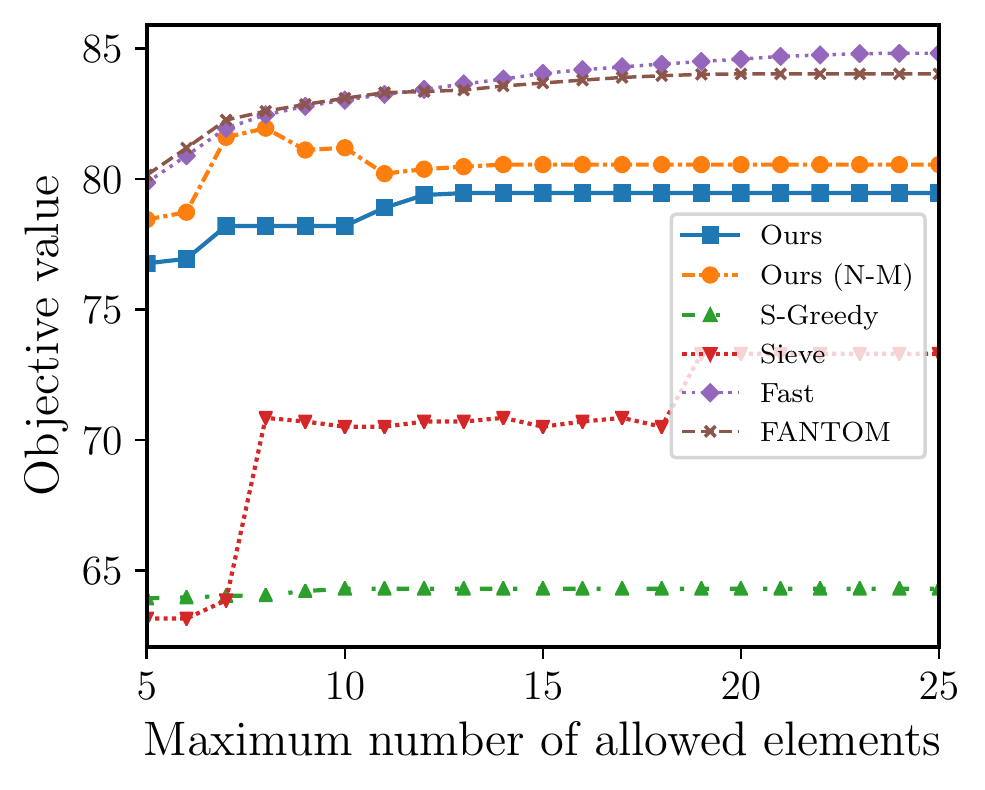}\label{fig:yelp-f-one}} 
	\subfloat[One Knapsack]	{\includegraphics[height=32mm]{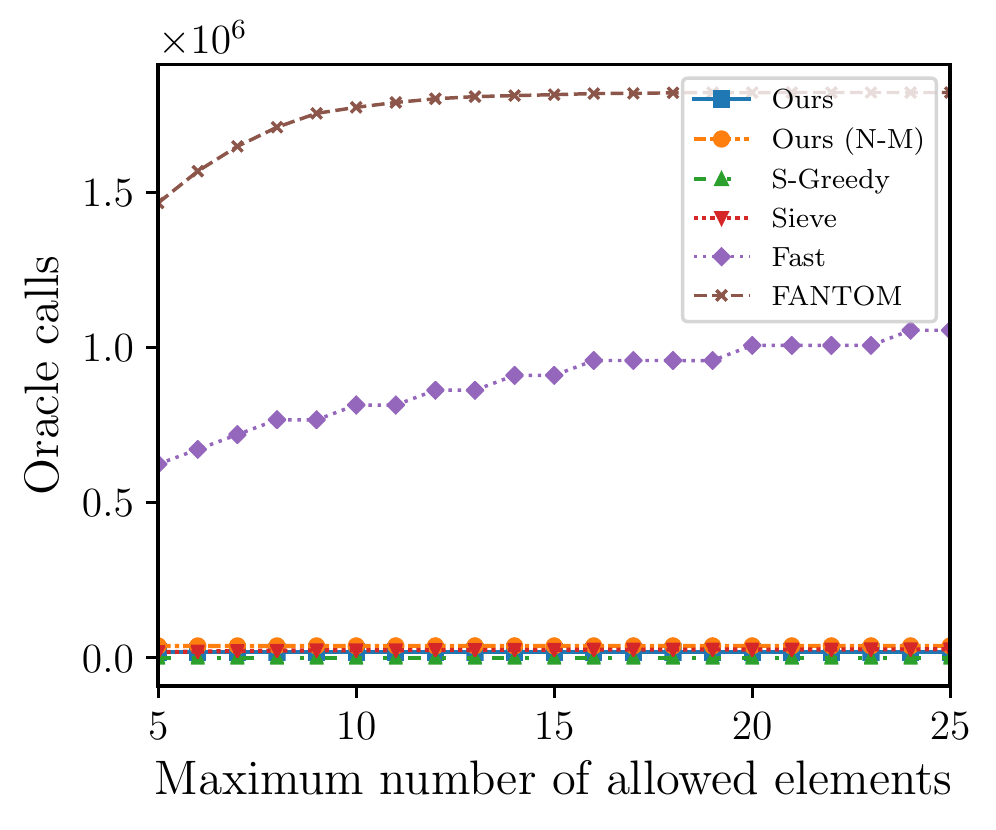}\label{fig:yep-o-one}}
	\subfloat[Two Knapsacks]	{\includegraphics[height=32mm]{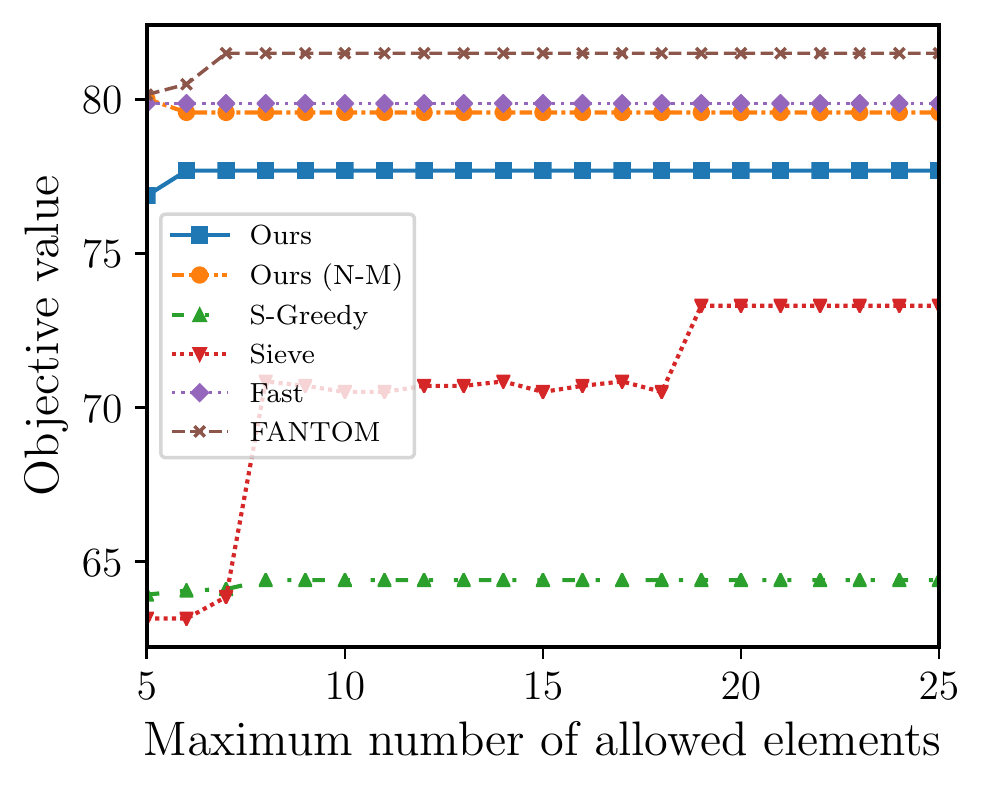}\label{fig:yelp-f-two}}
	\subfloat[Two Knapsacks]	{\includegraphics[height=32mm]{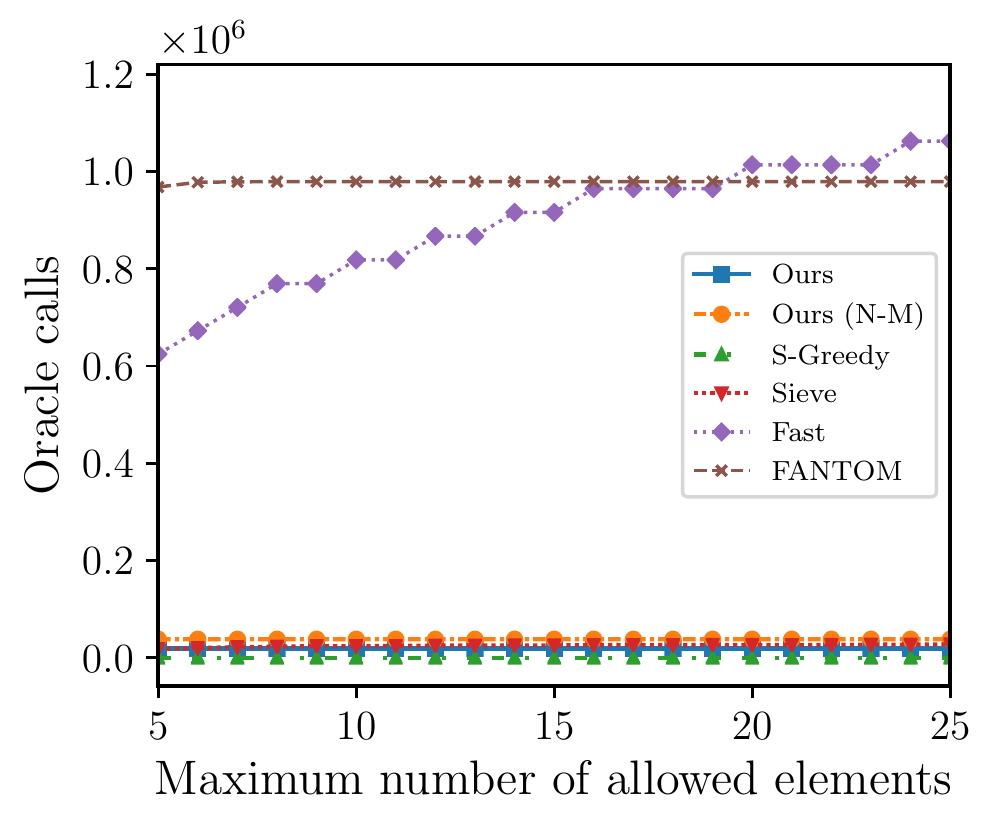}\label{fig:yelp-o-two}}
	\caption{Yelp Location Data Summarization}\label{fig:yelp}
\end{figure*}

\subsection{Twitter Summarization} \label{sec:twitter}
There are several news reporting Twitter accounts with millions of followers.
One interesting data summarization task is to provide a periodic synopsis of major events from the news feeds of these accounts. 
While finding an objective function to quantify the utility of a summary is a delicate task, the need to provide the summary in real-time for streams of data which are arriving in a fast pace makes the data summarization task even harder.

For this application, we use the twitter dataset provided in \cite{kazemi2019submodular}.
In order to cover the important events of the day without redundancy, we use a monotone and submodular function $f$ that encourages diversity in the selected set of tweets \cite{kazemi2019submodular}. Let us explain this function.
The function $f$ is defined over a ground set $\cN$ of tweets.
Assume that each tweet $u \in \cN$ consists of a non-negative value $\text{val}_u$ representing the number of retweets it has received and a set of $\ell_u$ keywords $W_u = \{ w_{u,1}, \cdots, w_{u, \ell_u}\}$ from the set of all possible keywords $\cW$.
The score of a word $w \in \cW $ for a given tweet $u$ is defined by
\[
	\text{score}(w,u) =
	\begin{cases}
		\text{val}_u & \text{if $w \in W_u$} \enspace,\\
		0 & \text{otherwise} \enspace,
	\end{cases}
\]
and the function $f$ is defined by 
\begin{equation*} 
f(S) = \sum_{w \in \cW} \sqrt{\sum_{u \in S} \text{score}(w,u)} \enspace .
\end{equation*}

Like in \cref{sec:yelp}, each one of our experiments involves a matroid constraint plus one or two knapsack constraints, which yields a $k$-extendible system constraint. The matroid constraint allows at most five tweets from each one of the six twitter accounts and at most $m$ tweets from all the accounts together. In the first knapsack constraint $c_1$, which is a constraint that is used in all the experiments of this section, the cost of each tweet is proportional to the absolute time difference (in months) between the tweet and the first of January 2019. In other words, we are more interested in tweets that are closer to the first day of the year 2019. 
We also have a second knapsack constraint $c_2$, which is used only in our second experiment. In this constraint, the cost of each element is proportional to the length (number of keywords) of the corresponding tweet, which enables us to provide shorter summaries.
We normalize the knapsack costs such that each unit of knapsack budget is equivalent to roughly 10 months for $c_1$ and 26 keywords for $c_2$, respectively. Then, we set the budgets of both knapsacks to $1$.

In \cref{fig:twitter-f-one,fig:twitter-o-one}, we observe the outcomes of different algorithms for the scenario with a single knapsack. It is evident that the utility of solutions returned by our proposed streaming algorithm exceeds the other baseline streaming algorithm. It is also interesting to point out that, for the case with two knapsack constraint, our streaming algorithms outperform even the Fast algorithm, which is one of the offline algorithms (see \cref{fig:twitter-f-two}).

\begin{figure*}[ht] 
	\centering  
	\subfloat[One knapsack] {\includegraphics[height=32mm]{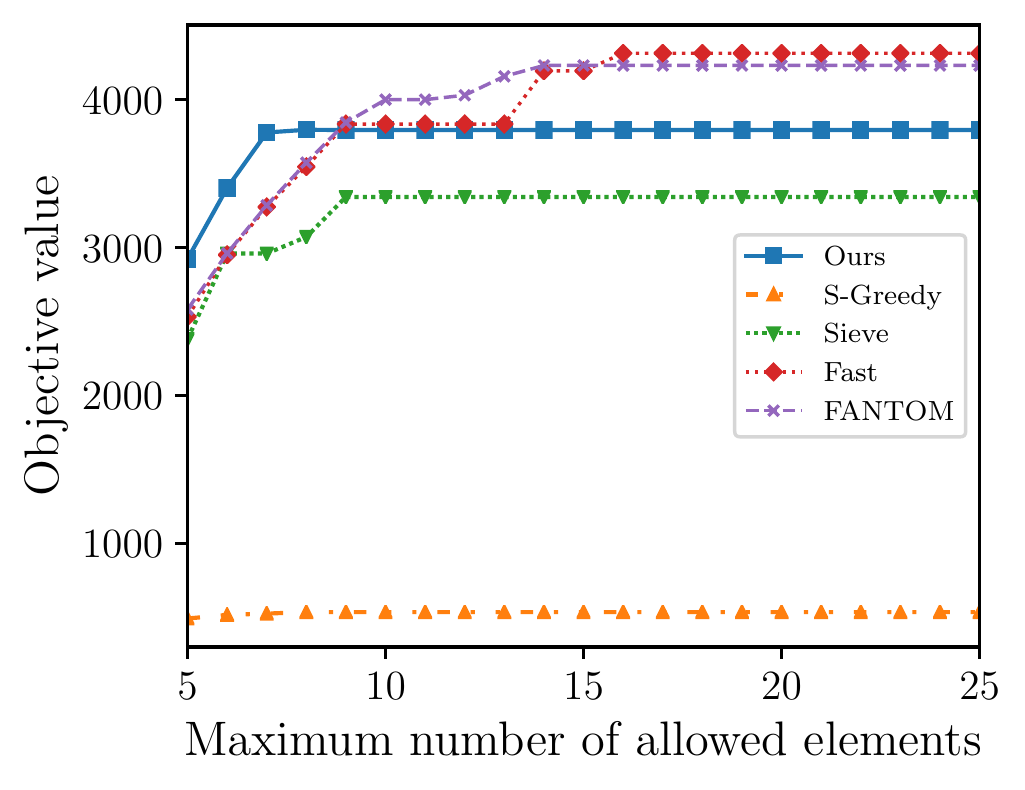}\label{fig:twitter-f-one}} 
	\subfloat[One Knapsack]	{\includegraphics[height=32mm]{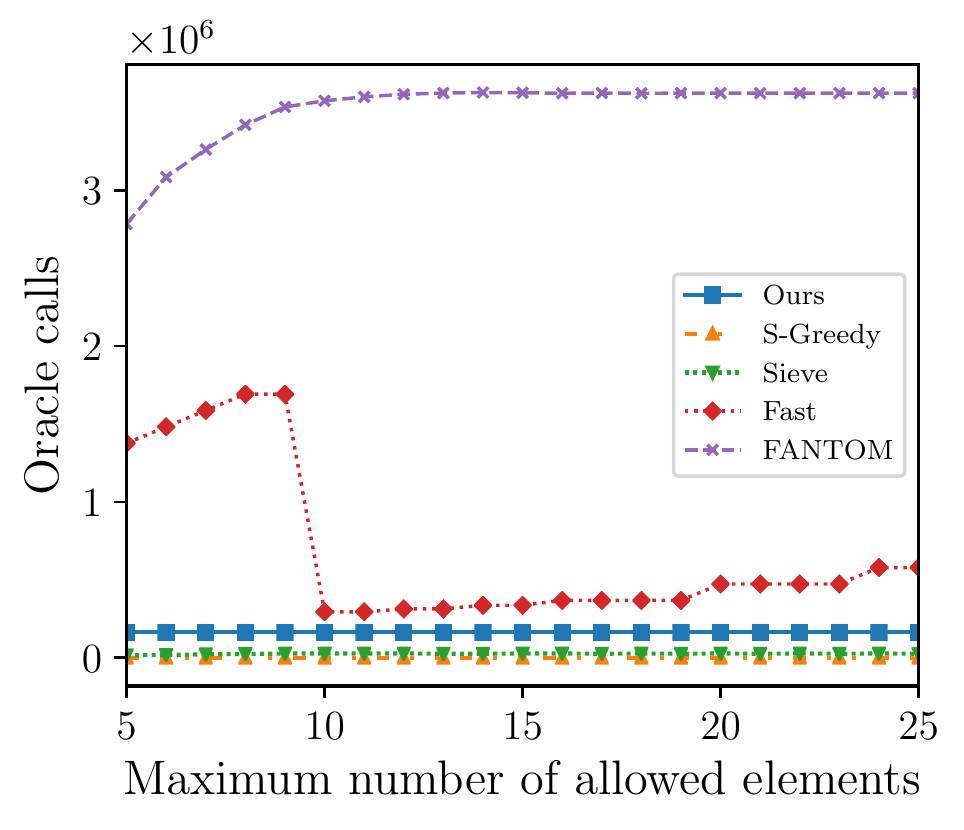}\label{fig:twitter-o-one}}
	\subfloat[Two Knapsacks]	{\includegraphics[height=32mm]{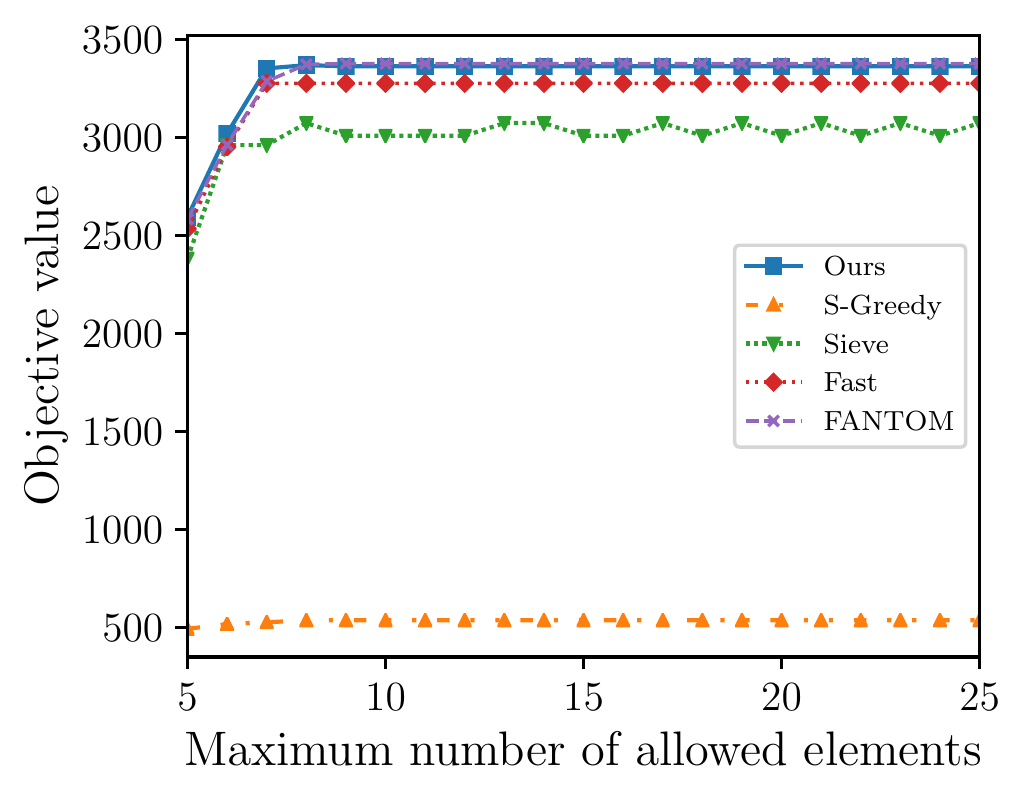}\label{fig:twitter-f-two}}
	\subfloat[Two Knapsacks]	{\includegraphics[height=32mm]{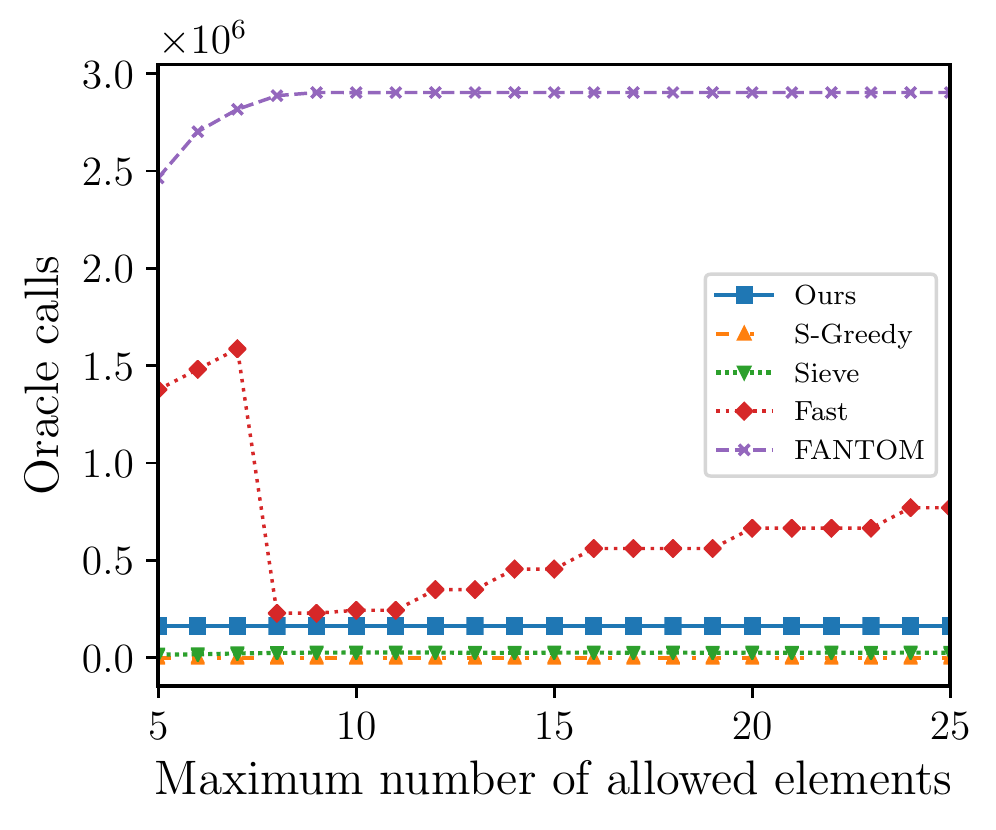}\label{fig:twitter-o-two}}
	\caption{Twitter Data Summarization: The maximum number of allowed tweets from each news agency in the summary is five. \label{fig:twitter}}
\end{figure*}

\section{Conclusion}

In this paper, we have proposed a novel framework for converting  streaming algorithms for  monotone submodular maximization into  streaming algorithms for non-monotone submodular maximization, which immediately led us to the currently tightest deterministic approximation ratio for submodular maximization subject to a $k$-matchoid constraint. We also proposed the first streaming algorithm for monotone submodular maximization subject to $k$-extendible and $k$-set system constraints, which (together with our proposed framework), yields approximation ratios of $O(k\log k)$ and $O(k^2\log k)$ for maximization of general non-negative submodular functions subject to the above constraints, respectively. Finally, we extensively evaluated the empirical performance of our algorithm against the existing work in a series of experiments including finding the maximum independent set in randomly generated graphs, maximizing linear functions over social networks, movie recommendation,  Yelp location summarization, and Twitter data summarization.

\bibliographystyle{plainnat}
\bibliography{streamingkextendible}
\appendix
\section{Counter Examples for Inequality~(\ref*{eq:prev-req-cond})} \label{appendix:counter-example}

In \cref{sec:framework}, we discussed the framework proposed by  \citet{mirzasoleiman2018streaming}  for maximizing a non-monotone submodular function using an algorithm for monotone functions.
This framework requires the input streaming algorithm to satisfy the inequality
\begin{align*}
f(S) \geq \alpha \cdot f(S \cup T) \enspace,
\end{align*}
where $S$ as the output of the algorithm, $T$ is an arbitrary feasible solution and $\alpha$ is a positive value. 
In the rest of this section, we provide two instances of the streaming maximization problem under a simple cardinality constraint $k$. These instances show that  the algorithms of~\cite{chekuri2015streaming}, \cite{chakrabarti2015submodular} and \cite{buchbinder2019online} fail to satisfy \cref{eq:prev-req-cond} for any constant $\alpha$.

Both our instances are based on a graph-cut function $f\colon 2^{V} \to \bR_{\geq 0}$ over vertices of a directed and weighted graph $G(V,E)$. This function is defined as follows:
\begin{align} \label{eq:cut-counter}
f(S) = \sum_{u \in S} \sum_{v \in V \setminus S} w_{u,v} \enspace,
\end{align}
where $w_{u,v} $ is the weight of the edge $e = (u,v)$.
It is easy to see that $f$ is a (usually non-monotone) submodular function. Furthermore, in our examples we assume the graph contains $3\maxcardinality + 1$ vertices named $u_0, u_1, u_2, \dotsc, u_{3\maxcardinality}$. The vertex $u_0$ does not appear in the input stream at all (it is there only for the purpose of allowing the description of the objective function as a cut function), and the other vertices appear in the stream in the order of their subscripts.

\subsection{Example for the Algorithms of Chekuri et al. and Chakrabarti and Kale} \label{sec:counter_chekuri}

The streaming algorithm of \citet{chekuri2015streaming}, in the context of a cardinality constraint, is given as \cref{alg:chekuri}. The algorithm of \citet{chakrabarti2015submodular} is very similar, and exhibits exactly the same behavior given the example we describe in this section, and therefore, we do not restate it here.

\begin{algorithm2e}[htb!]
	\DontPrintSemicolon
	\caption{Streaming Algorithm of \citet{chekuri2015streaming} } \label{alg:chekuri}
	$S \gets \emptyset$.\\
	\While{there are more elements in the stream}{
		$u \gets$ next element in the stream.\\
		\If{$|S| < \maxcardinality$}{
		\If{$f(u \mid S) \geq 0$}{
		$S \gets S \cup \{u\}$.\\
	}}
		\Else{
		$u' \gets \argmin_{x \in S} f(x: S)$, where $f(x : S)$ is the marginal contribution of $x$ to the part of $S$
		that arrived before $x$ itself.\\
			\If{$f(u \mid S) \geq 2 \cdot f(u' : S)$}{
			$S \gets (S \setminus \{u'\} )\cup \{u\}$.\\}
	}}
	\Return{$S$}.
\end{algorithm2e}

The counter example we suggest for \cref{alg:chekuri} is given by the weighted graph $G_1(V, E)$ shown in \cref{fig:counter_example}.
The weight of the black edges is 1, and the weight of the blue edges is $2 + \epsilon$ for some small and positive value $\epsilon$.

\begin{figure}[htb!] 
	\centering
	\includegraphics[width=4.in]{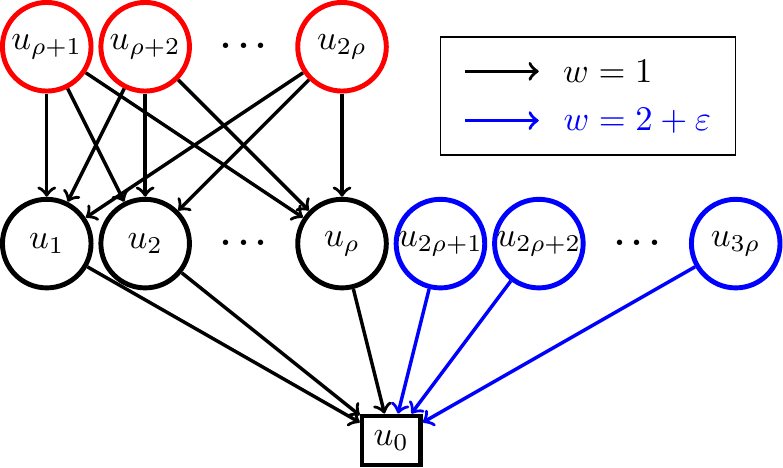}
	\caption{Weighted graph $G_1(V,E)$ used to define the counter example for the algorithm of \citet{chekuri2015streaming}.}\label{fig:counter_example_chekuri}
\end{figure}

\begin{lemma} \label{lemma:counter-example-chekuri}
	Assume $S$ is the output of Algorithm~\ref{alg:chekuri} for maximizing the graph-cut function $f$ (of the graph $G_1(V, E)$ and as defined in \cref{eq:cut-counter}) under a cardinality constraint $\maxcardinality$. Then,
	\[ f(S)  \leq \dfrac{2 + \epsilon}{\maxcardinality } \cdot  f(S \cup S^*) \enspace, \]
	where $S^*$ is the optimal solution.
\end{lemma}

\begin{proof}
First, it is clear that the optimal solution is the set $S^* = \{u_{\maxcardinality+1}, u_{\maxcardinality + 2}, \dots, u_{2\maxcardinality}\}$, for which $f(S^*) =  \maxcardinality^2$.
	When the first $\maxcardinality$ elements $V_1 = \{u_1, \dots, u_\maxcardinality\}$ arrive, all of them are added to the solution $S$ as the marginal gain of each one of them is $1$.
	Furthermore, when the elements $u\in S^*$ arrive, it is obvious that  $f(u \mid S) = 0$, and therefore,
	\[ f(u \mid S) < f(e' : S) = 1  \quad \forall u' \in S \enspace. \]
	Hence, none of the elements of $S^*$ would be added to the solution.
	Finally, it is straightforward to see that all elements in  $V_2 = \{u_{2\maxcardinality+1}, \dots, u_{3\maxcardinality} \}$ would replace an element in $V_1$ and  be in the final solution $S$. This is true because for $u \in V_2$ we have $f(u \mid S) = 2 + \epsilon$, which is larger than $2 \cdot f(u' : S )$ for $u' \in V_1$.
	The lemma now follows by observing that $f(S) = f(V_2) = (2 + \epsilon) \maxcardinality$ and $\maxcardinality^2 = f(S^*) \leq f(S \cup S^*)$.
\end{proof}

\subsection{Example for the Algorithm of Buchbinder et al.}

The streaming algorithm of \citet{buchbinder2019online} is given as Algorithm~\ref{alg:buchbinder}.

\begin{algorithm2e}[htb!]
	\DontPrintSemicolon
	\caption{Streaming Algorithm of \citet{buchbinder2019online} } \label{alg:buchbinder}
	$S \gets \emptyset$.\\
	\While{there are more elements in the stream}{
		$u \gets$ next element in the stream.\\
		\If{$|S| < \maxcardinality$}{
			\If{$f(u \mid S) \geq 0$}{
				$S \gets S \cup \{u\}$.\\
		}}
		\Else{
			$u' \gets \argmax_{x \in S} f(S \setminus \{x\} \cup \{u\} )$.\\
			\If{$f((S \setminus \{u'\}) \cup \{u\} ) - f(S) \geq \nicefrac{f(S)}{\maxcardinality}$}{
				$S \gets (S \setminus \{u'\}) \cup \{u\}$.\\}
	}}
	\Return{$S$}.
\end{algorithm2e}

In this section, the objective function of the counter example is given by the graph-cut function $f$ of the weighted graph $G_2(V, E)$ shown in \cref{fig:counter_example}. This graph has the same structure as the graph $G_1$ from \cref{sec:counter_chekuri}, but its weight selection is more involved. Specifically, in the  graph $G_2$, the weight of the black edges is 1 and there exist $\maxcardinality$ blue edges with weights $w_1, w_2, \dotsc, w_{\maxcardinality}$ given by $w_1 = 2$ and \[
	w_i =  \frac{2\maxcardinality +1 - i + \sum_{j=1}^{i-1} w_j}{\maxcardinality} \quad \forall\; i \geq 2
	\enspace.
	\]

\begin{figure}[htb!] 
	\centering
	\includegraphics[width=4.in]{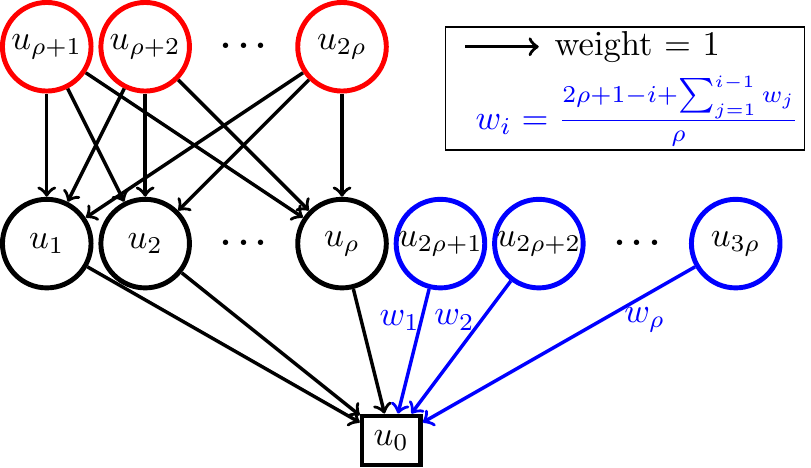}
	\caption{Weighted graph $G_2(V,E)$ used to define the counter example for the algorithm of \citet{buchbinder2019online}.}\label{fig:counter_example}
\end{figure}

\begin{lemma} \label{lemma:counter-example}
	Assume $S$ is the output of Algorithm~\ref{alg:buchbinder} for maximizing the graph-cut unction $f$ (of the graph $G_2(V, E)$ and as defined in \cref{eq:cut-counter}) under a cardinality constraint $\maxcardinality$.
	Then, for $\maxcardinality \geq 1 + e$,
	\[ f(S)  \leq \dfrac{e}{\maxcardinality} \cdot  f(S \cup S^*) \enspace, \]
	where $S^*$ is optimal solution.
\end{lemma}

\begin{proof}
	We begin the proof by showing, through an induction argument, that $w_i = 2 + \sum_{j=1}^{i-1} \binom{i-1}{j} \maxcardinality^{-j}$.
	The base of induction is trivial as $w_1 = 2$.
	Assuming the induction argument is correct for $h \leq i-1$, we prove that it is also correct for $i$.
	\begin{align*}
	w_i   = \frac{2k +1 - i + \sum_{j=1}^{i-1} w_j}{\maxcardinality} & = 2 + \frac{1 - i + \sum_{j=1}^{i-1} \left(2 + \sum_{\ell=1}^{j-1} \binom{j-1}{\ell} \maxcardinality^{-\ell} \right)}{\maxcardinality} \\
	& =  2 + \frac{i - 1 + \sum_{\ell =1}^{i-2} \sum_{j=\ell+1}^{i-1} \binom{j-1}{\ell} \maxcardinality^{-\ell}}{\maxcardinality} \overset{(a)}{=}  2 + \frac{i - 1 + \sum_{\ell =1}^{i-2} \binom{i-1}{\ell+1} \maxcardinality^{-\ell}}{\maxcardinality}  \\
	& = 2 + \binom{i-1}{1} \maxcardinality^{-1}+ \sum_{\ell =1}^{i-2} \binom{i-1}{\ell+1} \maxcardinality^{-(\ell+1)} = 2 + \sum_{j=1}^{i-1} \binom{i-1}{j} \maxcardinality^{-j} \enspace,
	\end{align*}
	where in $(a)$ we use the following well-known equality $\sum_{\ell=j}^{i} \binom{\ell}{j} = \binom{i+1}{j+1}$, which implies $\sum_{\ell=j + 1}^{i + 1} \binom{\ell - 1}{j} = \binom{i+1}{j+1}$. As a corollary of this proof, we get $w_i \leq 2 + [(1 + \maxcardinality^{-1})^{i-1} - 1] \leq 1 + (1 + \maxcardinality^{-1})^{\maxcardinality} \leq 1 + e$, which implies that the optimal solution is $S^* = \{e_{\maxcardinality+1}, e_{\maxcardinality+2}, \dots, e_{2\maxcardinality}\}$ whose value is $f(S^*) =  \maxcardinality^2$.
	
	When the first $\maxcardinality$ elements $V_1 = \{u_1, \dots, u_\maxcardinality\}$ arrive, all of them are added to the solution $S$ as the marginal gain of each one of them is $1$.
	Thus, when an element $u \in S^*$ arrive, we have  $ f(S \setminus \{u'\} \cup \{u\} ) - f(S) = 0$ for every $u' \in S$. Therefore,
	none of the elements of $S^*$ would be added to the solution.
	Next, we prove that all elements in  $V_2 = \{u_{2\maxcardinality+1}, \dots, u_{3\maxcardinality} \}$ would replace an element in $V_1$ and  be in the final solution $S$ of \cref{alg:buchbinder}.
	Again, we prove this claim by induction. 
	When $u_{2\maxcardinality+1}$ arrives, for all $u' \in S$ we have:
	\[  f(S \setminus \{u'\} \cup \{u_{2\maxcardinality+1}\} ) - f(S) = 1 \geq 1 = \frac{f(S)}{\maxcardinality} \enspace, \]
	and $u_{2\maxcardinality+1}$ replaces one of the elements from $V_1$.
	Assume now that elements $\{u_{2\maxcardinality+1}, \dots, u_{2\maxcardinality+i-1} \}$ for some integer $i < \maxcardinality$ have each replaced one of the element of $V_1$, and let us show that this implies that $u_{2\maxcardinality+i}$ would also replace one element $u'$ from $V_1$. This is true because for every such element $u' \in S$ we have
	\begin{align*}
		f(S \setminus \{u'\} \cup \{u_{2\maxcardinality+i}\} ) - f(S) = w_{i} - 1 = \frac{\maxcardinality + 1 - i + \sum_{j=1}^{i-1} w_j}{\maxcardinality} = \frac{\maxcardinality - ( i - 1) + \sum_{j=1}^{i-1} w_j}{\maxcardinality} = \frac{f(S)}{\maxcardinality} \enspace.
	\end{align*}
	As a corollary, we get that for the final solution $S = V_2$, we have  
	\begin{align*}
		f(S)
		={} &
		\sum_{i=1}^{\maxcardinality} w_i
		=
		2\maxcardinality + \sum_{i = 1}^\maxcardinality \sum_{j=1}^{i-1} \binom{i-1}{j} \maxcardinality^{-j}
		=
		2\maxcardinality + \sum_{i = 1}^{\maxcardinality - 1} \sum_{j=i}^{\maxcardinality-1} \binom{j}{i} \maxcardinality^{-i}\\
		={} &
		2\maxcardinality + \sum_{i=1}^{\maxcardinality-1} \binom{\maxcardinality}{i+1}\maxcardinality^{-i}
		=
		2 \maxcardinality + \maxcardinality \left((1 + \maxcardinality^{-1})^\maxcardinality - 2\right)
		\leq
		e\maxcardinality
		\enspace.
	\end{align*}
	This proves the lemma since $f(S^* \cup V_2) \geq f(S^*) \geq \maxcardinality^2$.
\end{proof}

\section{A Deterministic Streaming Algorithm for Submodular Maximization Subject to a \texorpdfstring{$k$}{k}-Matchoid Constraint} \label{sec:application}

As discussed in \cref{sec:framework}, \citet{chekuri2015streaming} already described a method to convert their algorithm for the problem of maximizing a non-negative monotone submodular function subject to a $k$-matchoid constraint into a deterministic algorithm that works also for non-monotone functions. The algorithm they obtained in this way has an approximation guarantee of $8k + \gamma$, where $\gamma$ is the approximation ratio of the offline algorithm used in the conversion. In this section we show that via our framework it is possible to get a somewhat better guarantee for the same problem.\footnote{Technically, the algorithm of~\cite{chekuri2015streaming} is identical to the algorithm obtained via our framework for $r = 2$, and the approximation guarantee they obtained can be reproduced using our framework by setting $r$ to this value. However, as our framework can handle other values of $r$ as well, we manage to get a slightly better guarantee by setting $r = 3$.}

The algorithm that we use as \AlgStream is the deterministic algorithm for monotone functions designed by~\cite{chekuri2015streaming}. Following we state some properties of this algorithm. We begin with a bound on its approximation guarantee. For this bound, let us denote by $S$ the final solution of the algorithm and by $A$ the set of elements that ever appeared in the solution maintained by the algorithm.
\begin{lemma}[Lemma~11 of~\cite{chekuri2015streaming}]
	Let $T \in \cI$ be an independent set. Then,
	\[ f(T \cup A) \leq \rho \alpha' + \dfrac{(1 + \beta')^2}{\beta'} \cdot k \cdot f(S) \enspace, \]
	where $\rho$ is an upper bound on the cardinality of the optimal set and the two non-negative parameters $\alpha'$ and $\beta'$ are inputs to the algorithm.
\end{lemma} 

In our notation, the last lemma implies that the deterministic algorithm of~\cite{chekuri2015streaming} is an $(k(1 + \beta')^2/\beta', \rho \alpha')$-approximation algorithm. \citet{chekuri2015streaming} also proved that this algorithm has the space complexity of a semi-streaming algorithm as long as $\alpha'$ is at least a constant fraction of $\Opt / \rho$. In particular, they showed the following lemma, which shows that in this regime the size of $A$ is linear in $\rho$.
\begin{lemma}[Lemma~5 of~\cite{chekuri2015streaming}]
	$|A| \leq \Opt / \alpha'$.
\end{lemma} 

For \AlgConstrained we use the {{\textsc{\textsc{RepeatedGreedy}}}\xspace} algorithm of~\cite{feldman2017greed}, which works for general $k$-systems constraints ($k$-matchoid constraints are a special case of $k$-systems constraints). The approximation ratio of this algorithm is $k + O(\sqrt{k})$, and it can be implemented to run in linear space. Plugging these two algorithms into our framework, we get the following corollary.

\begin{corollary} \label{cor:p-matchoid}
For every $\varepsilon \in (0, 1/8]$, by setting $\beta' = 1$, $\alpha' = \varepsilon \cdot \Opt / (3\rho)$ and $r = 3$, our framework produces a deterministic streaming algorithm for the problem of maximizing a non-negative (not necessary monotone) submodular function subject to a $k$-matchoid constraint. The approximation ratio of this algorithm is at most $(15/2 + 4\varepsilon)k + O(\sqrt{k})$.
\end{corollary}
\begin{proof}
By Theorem~\ref{theorem:non-monotone-streaming}, the algorithm obtained in this way produces a set whose value is at least
\[
	\frac{(r-1) \cdot \Opt - r\gamma}{r\alpha +  r(r-1)\beta/2}	
	=
	\frac{2 \cdot \Opt - \varepsilon \cdot \Opt}{12k +  3(k + O(\sqrt{k}))}	
	=
	\frac{2 \cdot \Opt - \varepsilon \cdot \Opt}{15k +  O(\sqrt{k})}	
	\enspace,
\]
and this implies that the approximation ratio of the algorithm is at most
\[
	\frac{15k + O(\sqrt{k})}{2 - \varepsilon}
	\leq
	\left(\frac{15}{2} + 4\varepsilon\right) \cdot k + O(\sqrt{k})
	\enspace.
	\qedhere
\]
\end{proof}

Before concluding this section, we note that the algorithm suggested by Corollary~\ref{cor:p-matchoid} assumes pre-knowledge of $\Opt$ and $\rho$ since these values are necessary for calculating $\alpha'$. It is possible to guess the value of $\Opt$ up to a small error using a technique originally due to~\cite{badanidiyuru2014streaming}, and this has no effect on the approximation guarantee of the algorithm (but slightly increases its space complexity). As the details of this are discussed by~\cite{chekuri2015streaming}, we avoid repeating them here. Regarding $\rho$, \citet{chekuri2015streaming} assumed pre-knowledge of $\rho$, and we take the same approach in this section. However, it is possible to modify the algorithm to avoid the need to have this pre-knowledge, and we demonstrate the technique leading to this possibility when discussing our algorithm for general $k$-systems.
\section{Extended Version of Our Algorithm} \label{app:full_algorithm}

In this section we present and analyze an extended version of our algorithm from Section~\ref{sec:algorithm} which need not assume pre-knowledge of $\maxcardinality$ and $\tau$. We do that in two steps. In Section~\ref{sec:extended_rho} we present a version of our algorithm that still assumes pre-access to $\tau$, but not to $\maxcardinality$; and in Section~\ref{sec:extended_tau} we show how to remove the need to known $\tau$ as well.

\subsection{Algorithm without Access to \texorpdfstring{$\rho$}{p}} \label{sec:extended_rho}

As an alternative to $\maxcardinality$, the algorithm we present in this section (which is given as \cref{alg:extended}) uses the size of a set $G$ produced by running the unweighted greedy algorithm on the entire input. Since the value of this alternative can increase over time, the algorithm has to create additional sets $E_i$ on the fly. We also note that the formula for $\ell$ used by \cref{alg:extended} is slightly different than the corresponding formula in \cref{alg:streaming-k-system}. 

\begin{algorithm2e}
	\DontPrintSemicolon
	\caption{Streaming Algorithm for $\ksys$-Systems (with no pre-access to $\maxcardinality$)} \label{alg:extended}
	\textbf{Input: } a value $\tau \in [M, 2M]$ and the parameter $k$ of the constraint.\\
	\textbf{Output: } a solution $T \in \cI$\\
	Let $G \gets \varnothing$, $\ell \gets -1$ and $h \gets \lceil \log_2(2k+1) \rceil$. \\
	\For{every element $u$ arriving}
	{
		\lIf{$G + u \in \cI$}{Add $u$ to $G$.}
		Let $\ell' \gets \lfloor 2\log_2 (k|G|) + 3 \rfloor$.\\
		\lFor{$i = \ell + 1$ \KwTo $\ell'$}{Initialize $E_i \gets \varnothing$.}
		Update $\ell \gets \ell'$.
		
		\BlankLine

		Let $m(u) \leftarrow f\left(u \mid \cup_{i=0}^{\ell} E_{i}\right) $.\\		
		\leIf{$m(u) > 0$}{Let $i(u) \leftarrow\left\lfloor\log _{2}(\tau / m(u))\right\rfloor$}{Let $i(u) \gets \infty$.}
		\lIf{$0 \leq i(u) \leq \ell $ and $E_{i(u)} + u \in \cI$}{Update $E_{i(u)} \gets E_{i(u)} + u$.}
	}
\BlankLine
\For{$j=0$ \textbf{to} $h - 1$}
{
	Let $i \gets j$ and $T_j \gets \emptyset$.\\
	\While{$i \leq \ell$}
	{
		\lWhile{there is an element $u \in E_i$ such that $T_j + u \in \cI$}{Update $T_j \gets T_j + u$.}
		$i \gets i + h$.
}
}
	\Return{the set $T$ maximizing $f$ among $T_0, T_1, \cdots, T_{h-1}$}.
\end{algorithm2e}

We begin the analysis of \cref{alg:extended} by showing that it has the space complexity of a semi-streaming algorithm.
\begin{lemma}
\cref{alg:extended} stores $O(\maxcardinality (\log \maxcardinality + \log k)) = \tilde{O}(\maxcardinality)$ elements at every given time point.
\end{lemma}
\begin{proof}
Observe that the set $G$ is kept as an independent set by the algorithm, and thus, its size is at most $\rho$, and we get that at all times $\ell = O(\log (k\rho)) = O(\log k + \log \maxcardinality)$. We now also note that \cref{alg:extended} stores elements only in the sets $E_0, E_1, \dotsc, E_\ell$ and the sets $T_0, T_1, \dotsc, T_{h - 1}$. Since these sets are kept independent by the algorithm, each one them contains at most $\maxcardinality$ elements. Thus, the number of elements stored by \cref{alg:extended} is upper bounded by
\[
	(\ell + h)\maxcardinality
	=
	[O(\log \maxcardinality + \log k)]\maxcardinality
	=
	O(\maxcardinality (\log \maxcardinality + \log k))
	\enspace.
	\qedhere
\]
\end{proof}

We now get to analyzing the approximation ratio of \cref{alg:extended}. One can verify that all the proofs in the analysis of the approximation ratio of~\cref{alg:streaming-k-system} from \cref{sec:algorithm} apply (as is) also to \cref{alg:extended}, except for the proof of Lemma~\ref{lem:E_bound}. Thus, in the rest of this section our objective is to show that Lemma~\ref{lem:E_bound} applies to \cref{alg:extended}, despite the fact that its original proof from \cref{sec:algorithm} does not apply to it.

Let us define $R$ to be a set including every element $u \in \cN$ for which either $i(u) < 0$ or $i(u) > \ell$ at the moment of $u$'s arrival. The following lemma allows us to bound the value of the elements in $R$.
\begin{lemma} \label{lem:R_bound}
For every independent set $S$, $\sum_{u \in S \cap R} m(u) \leq \tau/4$.
\end{lemma}
\begin{proof}
Let us denote the elements of $S \cap R$ by $u_1, u_2, \dotsc, u_r$ in the order of their arrival. For every $1 \leq j \leq |S \cap R|$, since $u_j \in R$, at the moment in which either $u_j$ arrived $i(u)$ was either negative or larger than $\ell$. However, since $S$ is independent, $\tau \geq M \geq f(\{u_j\}) \geq m(u_j)$, and thus, the first option cannot happen, which leaves us only with the case
\[
	\left\lfloor \log_2\left(\frac{\tau}{m(u_j)}\right) \right\rfloor \geq \ell + 1
	\Rightarrow
	\frac{\tau}{m(u_j)} \geq 2^{\ell + 1}
	\Rightarrow
	m(u_j) \leq \frac{\tau}{2^{\ell + 1}}
	\leq
	\frac{\tau}{2^{2\log_2(k|G|) + 3}}
	=
	\frac{\tau}{8k^2|G|^2}
	\enspace.
\]

We now observe that at the moment referred to by the previous paragraph the algorithm already received at least $j$ elements of $S$, and thus, the size of $G$ was at least $j/k$ (recall that the unweighted greedy algorithm is a $k$-approximation algorithm). Hence,
\[
	m(u_j)
	\leq
	\frac{\tau}{8j^2}
	\enspace.
\]
Summing up this inequality over all $1 \leq j \leq |S \cap R|$, we get
\[
	\sum_{u \in S \cap R} m(u)
	\leq
	\sum_{j = 1}^{|S \cap R|} \frac{\tau}{8j^2}
	\leq
	\frac{\tau}{8} \cdot \left[1 + \int_1^\infty \frac{dx}{x^2}\right]
	=
	\frac{\tau}{8} \cdot \left[1 - \left[\frac{1}{x}\right]_1^\infty\right]
	=
	\frac{\tau}{4}
	\enspace.
	\qedhere
\]
\end{proof}

Using the last lemma, we can now prove that Lemma~\ref{lem:E_bound} applies also to \cref{alg:extended}. Recall that $E = \cup_{i = 0}^\ell E_i$.
{
\renewcommand{\thetheorem}{\ref{lem:E_bound}}
\renewcommand{\theHtheorem}{repeat}
\begin{lemma}
For every set $S \in \cI$, $f(E \mid \varnothing) = \sum_{i=0}^{\ell} \sum_{u \in E_i} m(u) \geq \frac{f(S \cup E \mid \varnothing) - \tau/4}{2k+1}$.
\end{lemma}
\addtocounter{theorem}{-1}
}
\begin{proof}
	First, note that we have $f(E \mid \varnothing) = \sum_{i=0}^{\ell} \sum_{u \in E_i} m(u)$ because $m(u)$ is the marginal contribution of $u$ with respect to the elements that were added to $\cup_{i = 0}^\ell E_i$ before $u$. Let us also define, for every integer $0 \leq i \leq \ell $, $S_i = \{ u \in S \setminus R \mid i(u) = i \}$. Then,
	\begin{align*}
	f(E \mid \varnothing) = \sum_{i=0}^{\ell} \sum_{u \in E_i} m(u)  & \geq \sum_{i=0}^{\ell} |E_i| \cdot  \dfrac{\tau}{2^{i+1}} 
	\geq \dfrac{1}{k} \cdot \sum_{i=0}^{\ell} |S_i| \cdot \dfrac{\tau}{2^{i+1}} \\
	&\geq \dfrac{1}{2k} \cdot \sum_{i=0}^{\ell} \sum_{u \in S_i} m(u)
	= \dfrac{1}{2k} \cdot \left[\sum_{u \in S} m(u) - \sum_{u \in S \cap R} \mspace{-9mu}  m(u) \right] 
	\enspace,
	\end{align*}
	where the first and third inequalities hold since an element $u$ is added to a set $E_i$ only when $i = i(u)$, the second inequality holds since one can view $E_i$ as the output of running the unweighted greedy algorithm on a ground set which includes the independent set $S_i$ as a subset, and the last equality holds since $0 \leq i(u) \leq \ell$ for every element $u \not \in R$ because $\ell$ can only increase during the execution of \cref{alg:extended}.
	
	By the submodularity of $f$, we can immediately get
	\[
		\sum_{u \in S} m(u)
		\geq
		\sum_{u \in S} f(u \mid E)
		\geq
		f(S \mid E)
		=
		f(S \cup E \mid \varnothing) - f(E \mid \varnothing)
		\enspace.
	\]
	Combining the two above inequalities and the guarantee of Lemma~\ref{lem:R_bound} gives us
	\[
		f(E \mid \varnothing)
		\geq
		\frac{1}{2k} \cdot \left[f(S \cup E \mid \varnothing) - f(E \mid \varnothing) - \frac{\tau}{4} \right] 
		\enspace,
	\]
	and the lemma follows by rearranging this inequality.
\end{proof}

\subsection{Algorithm without Access to \texorpdfstring{$\tau$}{t}} \label{sec:extended_tau}

In this section we explain how to modify our algorithm so that it does not need to access to $\tau$. This modification is based on a technique due to~\cite{badanidiyuru2014streaming}, but it is made slightly more involved since we assume here no pre-knowledge of $\rho$, which is not the case in~\cite{badanidiyuru2014streaming}. Following is the crucial observation that we use in this section.
\begin{observation} \label{obs:not_taken}
Except for the sake of maintaining $G$, Algorithm~\ref{alg:extended} ignores an element $u$ if $\{u\} \not \in \cI$ or $f(\{u\}) \leq \tau / 2^{2\log_2 (k|G_u|) + 4}$, where $|G_u|$ is the size of $G$ immediately after the processing of $u$.
\end{observation}
\begin{proof}
The case of $\{u\} \not \in \cI$ is simple, so let us consider only the case $f(\{u\}) \leq \tau / 2^{2\log_2 (k|G_u|) + 4}$. Let $\ell_u$ be the value of $\ell$ immediately after the processing of $u$ by Algorithm~\ref{alg:extended}. The value $i(u)$ calculated by Algorithm~\ref{alg:extended} when processing $u$ obeys
\begin{align*}
	i(u)
	={} &
	\lfloor \log_2(\tau / m(u)) \rfloor
	\geq
	\lfloor \log_2(\tau / f(\{u\})) \rfloor\\
	\geq{} &
	\lfloor \log_2(2^{2\log_2 (k|G_u|) + 4} \rfloor
	=
	2\log_2 (k|G_u|) + 4
	\geq
	\ell_u + 1
	\enspace,
\end{align*}
where the first inequality follows from the submodularity of $f$, and the second inequality follows from the condition of the lemma.
\end{proof}

Using Observation~\ref{obs:not_taken} in mind, we now give the algorithm of this section as \cref{alg:no_tau}. This algorithm runs multiple copies of Algorithm~\ref{alg:extended}, each having a different $\tau$ values. The intuitive objective of the algorithm is to have a copy with $\tau = x$ for every $x$ which is a power of $2$, might belong to the range $[M, 2M]$ given the input so far, and might have accepted some element so far even given Observation~\ref{obs:not_taken}. Since the set $G$ is maintained by \cref{alg:extended} in a way which is independent of $\tau$, \cref{alg:no_tau} maintains $G$ itself, and we assume that the copies of \cref{alg:extended} that it creates use this set $G$ rather than maintaining their own set $G$.

\begin{algorithm2e}
	\DontPrintSemicolon
	\caption{Streaming Algorithm for $\ksys$-Systems (with no pre-access to $\maxcardinality$ and $\tau$)} \label{alg:no_tau}
	\textbf{Input: } the parameter $k$ of the constraint.\\
	\textbf{Output: } a solution $T \in \cI$\\
	Let $G \gets \varnothing$, $M' \gets -\infty$. \\
	\For{every element $u$ arriving}
	{
		\lIf{$G + u \in \cI$}{Add $u$ to $G$.}
		\lIf{$\{u\} \in \cI$}{Update $M' \gets \max\{M', f(\{u\})\}$.}
		Let $L = \{2^i \mid \text{$i$ is integer and } M' \leq 2^i \leq M' \cdot 2^{2\log_2 (k|G|) + 5}\}$.\\
		Delete any existing copy of \cref{alg:extended} whose $\tau$ value does not belong to $L$.\\
		\For{every $x \in L$}{Create a copy of \cref{alg:extended} with $\tau = x$, unless such a copy already exists.}
		Pass $u$ to all the copies of \cref{alg:extended} that currently exist.
	}
	\Return{the set maximizing $f$ among the output sets of all the currently existing copies of \cref{alg:extended}}.
\end{algorithm2e}

We begin the analysis of \cref{alg:no_tau} by analyzing its space complexity.
\begin{observation}
The number of elements stored by \cref{alg:no_tau} is larger than the number of elements stored by \cref{alg:extended} by a factor of $O(\log k + \log \rho)$.
\end{observation}
\begin{proof}
It suffices to show that \cref{alg:no_tau} maintains at most $O(\log k + \log \rho)$ copies of \cref{alg:extended} at any given time, and to do that it suffices to show that the size of the set $L$ created by \cref{alg:no_tau} in every iteration is upper bounded by $O(\log k + \log \rho)$. Note that the size of this set is at most
\begin{align*}
	\left\lceil \log_2\left(\frac{M' \cdot 2^{2\log_2 (k|G|) + 5}}{M'}\right)\right\rceil
	={} &
	\left\lceil \log_2\left(2^{2\log_2 (k|G|) + 5}\right)\right\rceil\\
	={} &
	\lceil 2\log_2 (k|G|) + 5\rceil
	\leq
	2\log_2 (k\rho) + 6
	=
	O(\log_2 k + \log_2 \rho)
	\enspace.
	\qedhere
\end{align*}
\end{proof}

Next, let us show that the approximation guarantee of \cref{alg:no_tau} is at least as good as the guarantee of \cref{alg:extended}.
\begin{lemma}
Let $S$ be the output set of \cref{alg:no_tau}. Then, $f(S)$ is at least as large as the value of the output set of \cref{alg:extended} when executed with some value $\tau \in [M, 2M]$.
\end{lemma}
\begin{proof}
Let $u$ be the first element to arrive which obeys both $\{u\} \in \cI$ and $f(\{u\}) \geq M / 2^{2\log_2 (k|G_u|) + 4}$. The set $L$ generated while processing this element necessarily includes a value $\bar{\tau}$ within the range $[M, 2M]$ because while processing $u$ we have
\[
	M' \cdot 2^{2\log_2 (k|G_u|) + 5}
	=
	f(\{u\}) \cdot 2^{2\log_2 (k|G_u|) + 5}
	\geq
	2M
	\enspace,
\]
and
\[
	M'
	=
	f(\{u\})
	\leq
	M
	\enspace.
\]
Moreover, we can observe that $\bar{\tau}$ belongs to any list $L$ generated after this point by \cref{alg:no_tau} because the value of $M'$ can only increase and the inequality $M' \leq M$ remains valid until the end of the algorithm. Thus, a copy of \cref{alg:extended} with $\tau = \bar{\tau}$ exists from the arrival of $u$ until \cref{alg:no_tau} terminates. Let us denote this copy by $C$.

Observation~\ref{obs:not_taken} and the definition of $u$ guarantee that the copy $C$ ignores the elements that arrived before $u$ if they are passed to it (except for the purpose of maintaining $G$, but we assume in this section that this work is done by \cref{alg:no_tau} itself rather than by the copies of \cref{alg:extended}). Thus, the output of $C$ is identical to the output it would have produced if all the elements had been passed to it, including elements that arrived before the creation of $C$. In other words, the output set of $C$ is the output set of \cref{alg:extended} when executed with some value $\tau = \bar{\tau} \in [M, 2M]$ on the entire input. Since $C$ survives until the end of the execution of \cref{alg:no_tau}, this implies that the output set of \cref{alg:no_tau} is at least as good as that.
\end{proof}
\section{Supplementary Experiments} \label{sec:sup_experiments}

In \cref{fig:planarity_linear_new}, we compare the performance of our streaming algorithm with the performance of Streaming Greedy and Sieve Streaming.
For the knapsack constraint, the cost of each edge $e = (u, v)$ is proportional to an integer picked uniformly at random from set the $\{1,2,3,4,5 \}$. 
The costs are normalized so that $\sum_{e \in E} c_e = |V|$, where $c_e$ represents the knapsack cost of edge $e$.
Like in~\cref{sec:planarity}, we observe that our streaming algorithm returns solutions with higher objective values for various knapsack budgets.

\begin{figure*}[ht] 
	\centering  
	\subfloat[Social graph]	{\includegraphics[height=32mm]{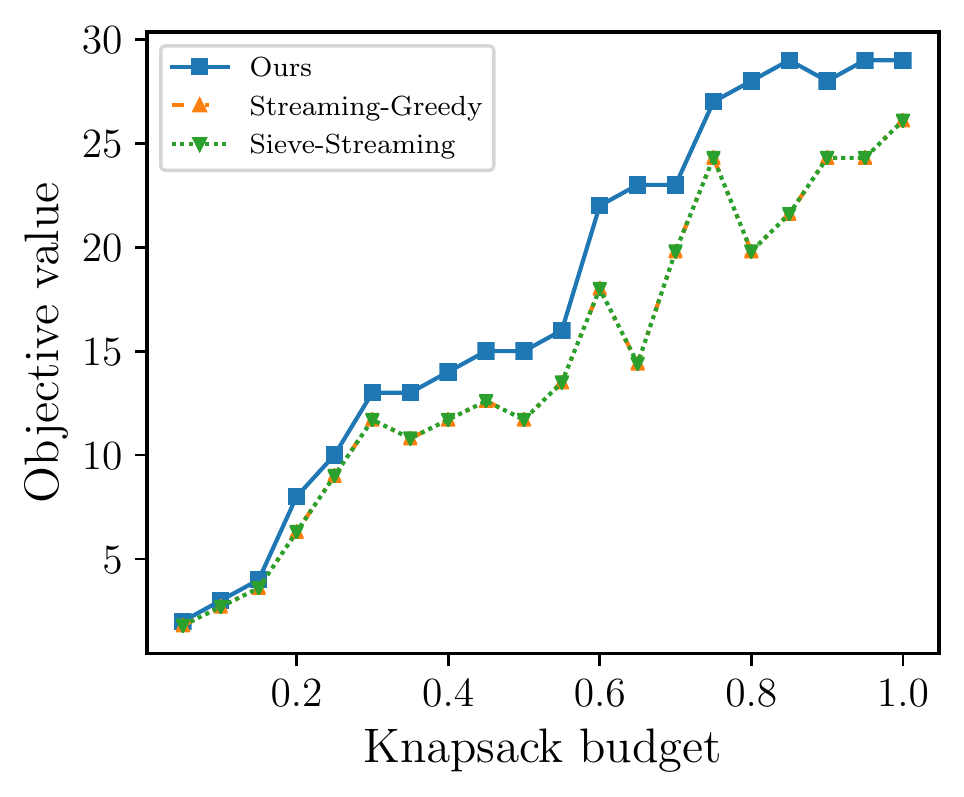}\label{fig:planar-0-new}}
	\subfloat[EU Email] {\includegraphics[height=32mm]{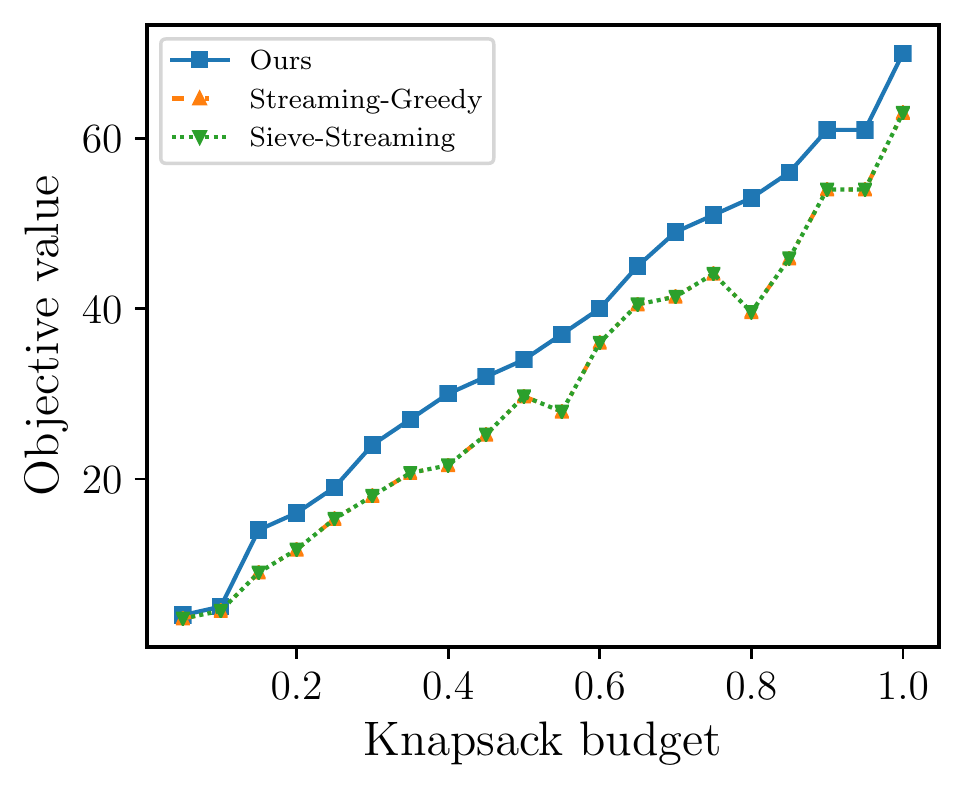}\label{fig:planar-EU-new}} 
	\subfloat[Facebook ego network]	{\includegraphics[height=32mm]{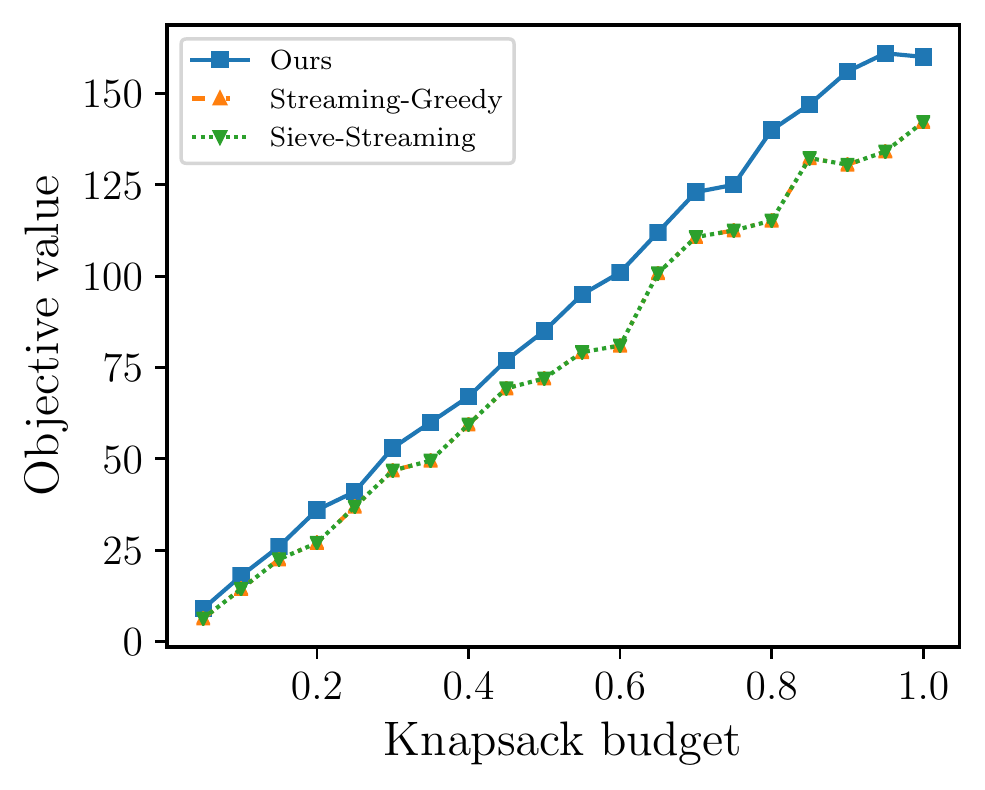}\label{fig:planar-1912-new}}
	\subfloat[Wiki vote network]	{\includegraphics[height=32mm]{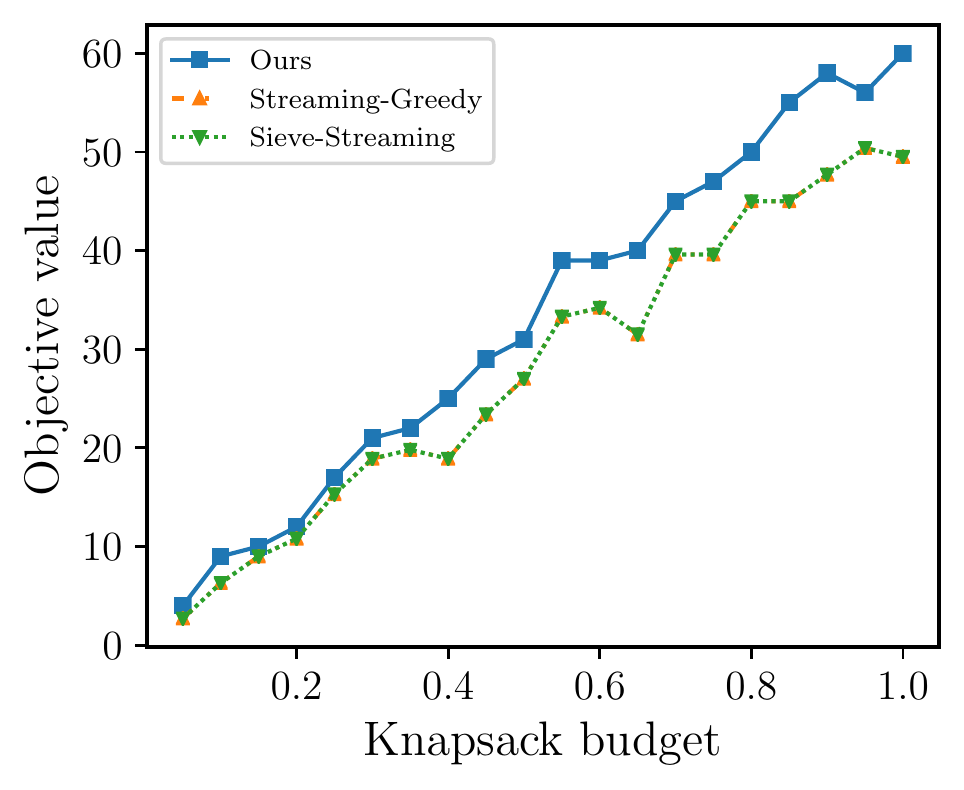}\label{fig:planar-wiki-new}}
	\caption{Planarity with knapsack (linear objective function).
		The weight of each edge is set to one.
		Knapsack cost of each edge $e = (u, v)$ is proportional to an integer picked uniformly at random from set the $\{1,2,3,4,5 \}$. The costs are normalized such that $\sum_{e \in E} c_e = |V|$, where $c_e$ represents the knapsack cost of edge $e$. We note that it is difficult to view the orange line in the above figures since it is mostly hidden behind the green line.}
	\label{fig:planarity_linear_new}
\end{figure*}
\end{document}